\documentclass[12pt]{birkjour}

\usepackage{graphicx}
\usepackage{amssymb}
\usepackage{epstopdf,comment,cite}
\usepackage{amsthm,  mathrsfs, enumerate}

\usepackage{amsmath,color}
\usepackage{MnSymbol}

\newtheorem{theorem}{Theorem}
\newtheorem{corollary}{Corollary}
\newtheorem{lemma}{Lemma}
\newtheorem{assumption}{Assumption}

\numberwithin{theorem}{section}
\numberwithin{lemma}{section}
\numberwithin{equation}{section}
\numberwithin{proposition}{section}
\numberwithin{corollary}{section}
\numberwithin{assumption}{section}

\newcommand{\req}[1]{Eq.\,(\ref{#1})}

\begin{document} 
\title{Phase Space Homogenization of Noisy Hamiltonian Systems}

\author[Birrell]{Jeremiah Birrell}
\address{Department of Mathematics\\
University of Arizona\\
Tucson, AZ, 85721, USA}
\email{jbirrell@math.arizona.edu}

\author[Wehr]{Jan Wehr}
\address{Department of Mathematics and Program in Applied Mathematics\\
University of Arizona\\
Tucson, AZ, 85721, USA}
\email{wehr@math.arizona.edu}



\date{\today}

\begin{abstract}
We study the dynamics  of an inertial particle coupled to forcing, dissipation, and noise in the small mass limit.  We derive an expression for the limiting  (homogenized) joint distribution of the position and (scaled) velocity degrees of freedom. In particular, weak convergence of the joint distributions is established, along with a bound on the convergence rate for a wide class of expected values.
\end{abstract}

\maketitle

\section{Introduction}
The motion of  a diffusing particle of non-zero mass, $m$, can be modeled by a stochastic differential equation (SDE) of the form
\begin{align}\label{model_sys}
dq_t=v_t dt,\hspace{2mm} m dv_t=-\gamma v_t dt+\sigma dW_t,
\end{align}
where $\gamma$ and $\sigma$ are the dissipation (or: drag) and diffusion coefficients respectively and $W_t$ is a Wiener process.   Smoluchowski \cite{smoluchowski1916drei} and Kramers \cite{KRAMERS1940284} pioneered the study of such diffusive systems in the small mass limit; see  \cite{Nelson1967} for a detailed discussion of the early literature. The field has since expanded far beyond  \req{model_sys} to more complicated models and settings, see for example \cite{doi:10.1137/S1540345903421076,Chevalier2008,bailleul2010stochastic,pinsky1976isotropic,pinsky1981homogenization,Jorgensen1978,dowell1980differentiable,XueMei2014,angst2015kinetic,bismut2005hypoelliptic,bismut2015}. 

In particular, more recently there has been interest in the phenomenon of {\em noise-induced drift}. An SDE can be derived that governs the dynamics of the state, $q^m_t$ (here and in the sequel we use a superscript to denote the $m$ dependence), in the limit $m\rightarrow 0$ and, when $\gamma$ ($\sigma$ if the Stratonovich integral is used) is state-dependent, the limiting equation can be shown to involve an additional drift term that was not present in the original system.   This was first derived in \cite{PhysRevA.25.1130} and has been studied in numerous subsequent works \cite{Sancho1982,volpe2010influence,Hottovy2014,herzog2015small,particle_manifold_paper,BirrellHomogenization}.  See  \cite{Hottovy2014} for further  references and discussion.  See also \cite{FritzRoughPath} for a rough paths perspective on the singular nature of the small mass limit. 

 Such problems can be classified under the broad umbrella of homogenization, for which \cite{fouque2007wave} and   \cite{pavliotis2008multiscale}   are recent  sources.  In this paper, we study the small mass limit of the joint distribution of $q^m_{t_1},...,q^m_{t_N},\sqrt{m} v^m_{t_1},...,\sqrt{m} v^m_{t_N}$.  In particular, we prove that that the $\sqrt{m}$ scaling of the velocity is the correct one to produce a nontrivial weak limit.  This is a generalization of previous results referenced above, which considered only the limit of the state variables $q_t^m$.   

In fact, a crucial step of the proof of the main result in \cite{Hottovy2014} consists of showing that the kinetic energy $m\|v^m_t\|^2$ is of order one.  This upper bound strongly suggests that $v^m_t$ diverges as ${1 \over \sqrt{m}}$.  In the present work we prove a much more detailed statement, which identifies the limiting distribution of the scaled velocity $\sqrt{m}v^m_t$.  This provides a more precise picture of homogenization and provides a tool to study the behavior of physically important quantities such as entropy production.  

Our results hold for state-dependent, matrix-valued drag and diffusion, even if the fluctuation dissipation relation is not satisfied, and therefore there is no notion of local temperature. We will show that there is still a notion of local equilibrium that describes the limit of the scaled velocity process in this case.

\subsection{Prior Results}

Generalizing \req{model_sys} to allow for  time and state-dependent drag, noise, and external forcing, we arrive at the type of Langevin equation that will be studied in this work
\begin{align}
dq_t^m=&v_t^m dt,\label{q_eq0}\\
md(v^m_t)_i=&\left(-\tilde \gamma_{ik}(t,q_t^m) (v_t^m)^k+  F_i(t,x^m_t)\right)dt+\sigma_{i\rho}(t,q_t^m)dW^\rho_t,\label{v_eq}
\end{align}
where $x_t^m\equiv (q_t^m,p_t^m)$, $\tilde\gamma$, $F$, and $\sigma$ are continuous, and $\tilde\gamma$ consists of a symmetric drag matrix, $\gamma$, and an antisymmetric part (magnetic field) generated by a $C^2$ vector potential, $\psi$:
\begin{align}\label{tilde_gamma_def}
\tilde \gamma_{ik}(t,q)\equiv\gamma_{ik}(t,q) +\partial_{q^k}\psi_i(t,q)-\partial_{q^i}\psi_k(t,q).
\end{align}

Stated in the framework  of Hamiltonian systems, our results cover Hamiltonians of the form
\begin{align}\label{hamiltonian_family}
H(t,q,p)=\frac{1}{2 m}\|p-\psi(t,q)\|^2+V(t,q)
\end{align}
where  $\psi$ represents the vector potential of an electromagnetic field (charge set to one) and the $C^2$ function $V$ represents an (electrostatic) potential.  The relation to \req{q_eq0}-\req{v_eq}  is as follows:\\
The canonical momentum, $p_t^m$, is related to the velocity, $v_t^m$, by $mv_t^m= p_t^m-\psi(t,q_t^m)$.   The total forcing is
 \begin{align}
F(t,x)=-\partial_t\psi(t,q)-\nabla_q V(t,q)+\tilde F(t,x),
\end{align} 
where $\tilde F$ is any additional external forcing.  We will think of equations \req{q_eq0}-\req{v_eq} as coming from an electromagnetic Hamiltonian in this way, but because we allow for a non-Hamiltonian external forcing $\tilde F$, this is largely a stylistic choice.

It will be convenient to make the definition
\begin{align}\label{u_def}
u_t^m=mv_t^m.
\end{align}
Note that $v_t^m=O(m^{-1/2})$ translates into $u_t^m=O(m^{1/2})$. Using \req{u_def}, we rewrite the system as
\begin{align}
dq_t^m=&\frac{1}{m }u_t^m dt,\label{q_eq}\\
d(u^m_t)_i=&\left(-\frac{1}{m}\tilde \gamma_{ik}(t,q_t^m) (u_t^m)^kdt+  F_i(t,x^m_t)\right)dt+\sigma_{i\rho}(t,q_t^m)dW^\rho_t.\label{u_eq}
\end{align}
In \cite{BirrellHomogenization} we showed that, for a large class of such systems, there exists unique global in time solutions $(q_t^m,u_t^m)$ that converge to $(q_t,0)$ as $m\rightarrow 0$, where  $q_t$ is the solution to a certain limiting SDE.  We summarize the precise mode of convergence in Assumption \ref{assump:conv} below, which we take as the  starting point for this work.  See Appendix \ref{app:assump} for a  list of properties that guarantee that the following  holds.  
\begin{assumption}\label{assump:conv}
 For any  $T>0$, $p>0$  we have
\begin{align}\label{results_summary1}
\sup_{t\in[0,T]} E\left[\|u_t^m\|^p\right]^{1/p}=O(m^{1/2}),\hspace{2mm} \sup_{t\in[0,T]}E\left[\|q_t^m-q_t\|^p\right]^{1/p}=O(m^{1/2}) 
\end{align}
as $m\rightarrow 0$, where $q_t$ is the solution to an SDE of the form 
\begin{align}\label{q_SDE}
dq_t=& \tilde \gamma^{-1}(t,q_t)F(t,q_t,\psi(t,q_t))dt+S(t,q_t)dt+\tilde \gamma^{-1}(t,q_t)\sigma(t,q_t) dW_t.
\end{align}
$S(t,q)$ is called the {\em noise induced drift}, see \cite{BirrellHomogenization}, and is given by
(employing the summation convention on repeated indices):
\begin{enumerate}
\item $S^i(t,q)\equiv  \partial_{q^k}(\tilde\gamma^{-1})^{ij}(t,q) \delta^{kl}G_{jl}^{rs}(t,q)\Sigma_{rs}(t,q)$,
\item $G_{ij}^{kl}(t,q)\equiv \delta^{rk}\delta^{sl}\int_0^\infty (e^{-\zeta \tilde\gamma(t,q)})_{ir} (e^{-\zeta \tilde\gamma(t,q)})_{js} d\zeta$,
\item $\Sigma_{ij}\equiv\sum_\rho\sigma_{i\rho}\sigma_{j\rho}$.
\end{enumerate}

The initial conditions are assumed to satisfy $E[\|q^m_0\|^p]<\infty$, $E[\|q_0\|^p]<\infty$, and $E[\|q_0^m-q_0\|^p]^{1/p}=O(m^{1/2})$ for all $p>0$.  $q_t$ also satisfies the bound
\begin{align}\label{q_Lp_bound}
E\left[\sup_{t\in[0,T]}\|q_t\|^p\right]<\infty
\end{align}
for all $T>0$, $p>0$.
\end{assumption}
Note that we define the components of $\tilde\gamma^{-1}$ so that 
\begin{align}\label{tilde_gamma_inv_def}
(\tilde\gamma^{-1})^{ij}\tilde\gamma_{jk}=\delta^i_k,
\end{align}
\begin{comment}
i.e. to think of them as linear maps, we lower the second index on $\tilde\gamma^{-1}$ and raise the first index on $\tilde\gamma$
\end{comment}
and for any $v_i$ we define the contraction $(\tilde\gamma^{-1}v)^i=(\tilde\gamma^{-1})^{ij}v_j$.

As stated above, a comprehensive list of assumptions that guarantee the  convergence and boundedness properties from Assumption \ref{assump:conv} can be found in Appendix \ref{app:assump}. Several of them (and their consequences) will be explicitly used in the remainder of this paper and so we restate them here.
\begin{assumption}\label{assump:bounds}

For any $T>0$:
\begin{enumerate}
\item $\tilde \gamma$ is Lipschitz in $(t,q)$ for $(t,q)\in [0,T]\times\mathbb{R}^n$.
 \item $ \tilde \gamma$, $ \tilde \gamma^{-1}$, $\sigma$, and $\tilde F$, $\nabla_q V$,  $\partial_t\psi$, and $S$ are bounded on $[0,T]\times\mathbb{R}^n$.
\item  $\|u^m_0\|^2\leq C m$ for some $C>0$.
\item $\gamma$ is symmetric with eigenvalues bounded below by a constant $\lambda>0$  on $[0,T]\times\mathbb{R}^n$.  Note that this implies the real parts of the eigenvalues of $\tilde\gamma$ are also bounded below by $\lambda$. (See Lemma \ref{eig_bound_lemma1}.)
\end{enumerate}
We make the additional assumptions, not needed in \cite{BirrellHomogenization}, that
\begin{enumerate}
  \setcounter{enumi}{4}
\item $\Sigma\equiv\sigma\sigma^T$ has eigenvalues bounded below by $\mu>0$ on $[0,T]\times\mathbb{R}^n$.
\item $\sigma$ is Lipschitz in  $(t,q)\in [0,T]\times\mathbb{R}^n$.
\end{enumerate}
\end{assumption}
The boundedness assumptions we make can likely be relaxed by using the techniques developed in \cite{herzog2015small} and employed in \cite{BirrellHomogenization}, but we don't pursue that here.

Our calculation will use a martingale representation result which requires us to be precise about some properties of the probability space on which our equations are formulated.
\begin{assumption}\label{assump:prob_space}
Given a Wiener process, $W_t$, on some probability space $(\Omega,\mathcal{H},P)$, let $\mathcal{F}^W_t$ be the natural filtration of $W_t$ and $\mathcal{C}$ be any  sub $\sigma$-algebra of $\mathcal{H}$ that is independent of $\mathcal{F}^W_\infty$. Define $\mathcal{G}^{W,\mathcal{C}}_t\equiv \sigma(\mathcal{F}^W_t\cup\mathcal{C})$ and complete it with respect to $(\mathcal{G}^{W,\mathcal{C}}_\infty,P)$ to form $\overline{\mathcal{G}^{W,\mathcal{C}}_t}$.  Note that $(W_t,\overline{\mathcal{G}^{W,\mathcal{C}}_t})$ is still a Wiener process on $(\Omega,\overline{\mathcal{G}^{W,\mathcal{C}}}_\infty,P)$ and this space satisfies the usual conditions \cite{karatzas2014brownian}. 

 For the remainder of this paper, it is assumed that we work within a filtered probability space
\begin{align}\label{prob_space_def}
 (\Omega,\mathcal{F},\mathcal{F}_t,P)\equiv (\Omega, \overline{\mathcal{G}^{W,\mathcal{C}}}_\infty,\overline{\mathcal{G}^{W,\mathcal{C}}_t} ,P)
\end{align}
constructed in the manner described above from some Wiener process, $W_t$, (the same one used to formulate the SDEs) and some sub sigma-algebra, $\mathcal{C}$, independent of $W$.
\end{assumption}

\subsection{Fundamental Solution}

A key tool in this paper will be the fundamental solution to the equation
\begin{align}\label{Phi_eq}
\frac{d}{dt}\Phi_t^m=-\frac{1}{m}\tilde \gamma(t,q^m_t)\Phi^m_t,\hspace{2mm} \Phi^m_0=I.
\end{align}
We alert the reader that here, and in the sequel, the superscript $m$ on matrix-valued quantities denotes the value of the mass (similar to the vector-valued quantities $q_t^m$ etc.) and not a power.  We also emphasize that \req{Phi_eq} is solved pathwise; for a fixed realization of $q_t^m$, where $q_t^m$ comes from the solution of \req{q_eq}-\req{u_eq}, the  ODE \req{Phi_eq} is solved, thus defining a $C^1$ process $\Phi^m_t$.

The symmetric part of $ \tilde \gamma$ is $\gamma$, which has eigenvalues bounded below by $\lambda>0$.  Hence (see, for example, p.86 of \cite{teschl2012ordinary}), for $t\geq s$, we have the important bound
\begin{align}\label{fund_sol_decay}
\|\Phi^m_t(\Phi^m_s)^{-1}\|\leq e^{-\lambda(t-s)/m}.
\end{align}
Note that, for this bound, it is important that the $\ell^2$ matrix norm is used.

\req{u_eq} is a linear equation for $u_t^m$, so $\Phi^m$ furnishes us with the explicit solution in terms of $q_t^m$:
\begin{align}\label{u_sol}
u^m_t=\Phi^m_t\left(u^m_0+\int_0^t (\Phi^m_s)^{-1}F(s,q^m_s) ds+\int_0^t(\Phi^m_s)^{-1} \sigma(s,q^m_s) dW_s\right).
\end{align}

It will also be important to recall the following decomposition of the stochastic convolution term (Lemma 5.1 in \cite{particle_manifold_paper}).
\begin{align}\label{convol_decomp}
&\Phi^m_t \int_0^t(\Phi^m_s)^{-1} \sigma(s,q^m_s) dW_s\\
=&\Phi^m_t \int_0^t \sigma(s,q^m_s) dW_s+\frac{1}{m}\Phi^m_t\int_0^t(\Phi^m_s)^{-1} \tilde \gamma(s,q^m_s) \int_s^t \sigma(r,q^m_r) dW_r ds.\notag
\end{align}

\subsection{Summary of Results}\label{sec:summary}
Define
\begin{align}\label{z_def}
z_t^m=u_t^m/\sqrt{m}=\sqrt{m}v_t^m.
\end{align}
 Our primary result is an expression for the limiting joint distribution of the random variables $q_{t_1}^m,...,q_{t_N}^m,z_{t_1}^m,...z_{t_N}^m$, $0<t_1<...<t_N$,  as $m\rightarrow 0$  in terms of the limiting position process $q_t$.  

More generally, we will consider random variables 
\begin{align}\label{Y_def}
Y^m\equiv (J^m, q_{t_1}^m,...,q_{t_N}^m, z_{t_1}^m,...,,z_{t_N}^m),
\end{align} where $J^m$ are $\mathbb{R}^d$-valued $L^2$ random variables that satisfy:\\
\begin{assumption}\label{assump:J}
We assume that $J^m$ converges in $L^2$ to  some $\mathbb{R}^d$-valued random variable $J$.
\end{assumption}
\noindent We are purposely general here, having in mind certain processes constructed from the $q_t^m$'s, such as $J_t^m=\int_0^t g(r,q_r^m)dr$. Such processes appear, for example, in the computation of entropy production \cite{Chetrite2008,gawedzki2013fluctuation,Birrell_entropy}.

As part of our main result, Theorem \ref{thm:conv_dis}, we will prove  that the distributions of  $(J^m,q^m_{t_1},...,q^m_{t_N},z_{t_1}^m,...,z_{t_N}^m)$ converge weakly to 
\begin{align}
d\nu=&  \left(\prod_{j=1}^N\frac{1}{(2\pi)^{n/2} |\det M(t_j,q_j)|^{1/2}}\exp\left[-z_j\cdot M^{-1}(t_j,q_j)z_j/2\right]dz_j\right)\\
&\times \mu_{J,t}(dJ,dq_1,...,dq_N)\notag
\end{align}
as $m\to 0$, where $q_t$ is the solution to the limiting SDE, \req{q_SDE}, $\mu_{J,t}$ is the distribution of $(J,q_{t_1},...q_{t_N})$, and
\begin{align}\label{M_def}
M(t,q)=\int_0^\infty  e^{-\tilde\gamma(t,q) \zeta}\Sigma(t,q) e^{-\tilde\gamma^T(t,q) \zeta}d\zeta.
\end{align}
(Recall that $\Sigma=\sigma\sigma^T$). For $N=1$, we also derive a bound on the convergence rate of $E[h(Y^m)]$ for a wide class of functions $h$.  See Section \ref{sec:conv_rate}, especially Theorem \ref{thm:conv_dis2}.

In Corollary \ref{corr:fluc_dis} we give a special case of this result, namely if a fluctuation dissipation relation holds pointwise for a time and position dependent ``temperature" $T(t,q)$, i.e.
\begin{align}
\Sigma_{ij}(t,q)=2k_BT(t,q) \gamma_{ij}(t,q),
\end{align}
then the limiting distribution is
\begin{align}
d\nu = \left(\prod_{j=1}^N\left(\frac{\beta(t_j,q_j)}{2\pi}\right)^{n/2}\exp\left[- \beta(t_j,q_j) \|z_j\|^2/2\right]dz_j\right) \mu_{J,t}(dJ,dq_1,...,dq_N).
\end{align}
Here we recognize the Gibbs measure for the $z$-variables, and so we can interpret  this corollary as expressing an instantaneous equilibration of the scaled momentum variables (in particular, of the kinetic energy)  in the limit $m\to 0$.  We also see that there an asymptotic independence of the (rescaled) momentum and past history of the position.  In particular, if the temperature (or more generally $M$ from \req{M_def}) doesn't depend on $q$, then the $z_t^m$'s converge to independent Gaussians i.e. uncorrelated in time and independent of the $q_t$'s as well.   Note that, although $q_t=\int_0^t v_sds$, this does not imply that $q_t$ is a Wiener process.  See \cite{Birrell_entropy} for further information on such integral processes.

\section{Computing the Limiting Distribution}\label{sec:limit_dist}
Fix $N>0$, $0<t_1<...<t_N$ and let $Y^m$ be defined by \req{Y_def}. Weak convergence of the distributions of $Y^m$ will be proven by showing pointwise convergence of the Fourier transforms of $Y^m$.

Before considering the general case, we consider the much simpler situation where the forcing, $F$, and initial velocity, $v_0^m$, vanish and $\tilde\gamma$ and $\sigma$  do not depend on $q$ or $t$. In addition to being a more transparent illustration of some of our methods, we will see in Section \ref{sec:simp_seq} that the full result can be reduced to this case.
\subsection{Time and State-Independent Case}\label{sec:state_ind}
In the case of constant $\tilde\gamma$, $\sigma$ the solution of \req{Phi_eq} is simply a matrix exponential $\Phi_t^m=e^{-t\tilde\gamma/m}$ and the formulas \req{u_sol} and \req{z_def} for the scaled velocity process can be written
\begin{align}\label{z_sol_simple}
z^m_t=\frac{1}{\sqrt{m}}e^{-t\tilde\gamma/m}\int_0^t e^{s\tilde\gamma /m} \sigma dW_s.
\end{align}
In particular, it does not depend on the position process, $q_t^m$.

We now compute the limit of its Fourier transform as $m\to 0$.  In fact, it will be useful  to derive a more general result, namely the small mass limit of quantities of the form
\begin{align}
 \tilde H_t^m\equiv E\left[\tilde h\left(\frac{1}{\sqrt{m}}  e^{-\tilde\gamma t/m} \int_0^t e^{\tilde\gamma s/m} \sigma  dW_s\right)\right],
\end{align}
where  $\tilde h:\mathbb{R}^n\rightarrow\mathbb{C}$ is a  polynomially bounded continuous function. Recall that we are assuming that the symmetric part of $\tilde \gamma$ has eigenvalues bounded below by $\lambda$ and $\Sigma\equiv\sigma\sigma^T$ has eigenvalues bounded below by $\mu$.  For our subsequent purposes, it is very important that we be explicit about how our bounds depend on these, and other, constants.

The process $\frac{1}{\sqrt{m}}  e^{-\tilde\gamma t/m} \int_0^t e^{\tilde\gamma s/m} \sigma  dW_s$ is normally distributed with mean zero and covariance
\begin{align}
M_{m,t}=&\frac{1}{m}\int_0^t (e^{-\tilde\gamma (t-s)/m} \sigma)(e^{-\tilde\gamma (t-s)/m} \sigma)^Tds\\
=&\frac{1}{m}\int_0^t  e^{-\tilde\gamma (t-s)/m} \Sigma e^{-\tilde\gamma ^T (t-s)/m}ds,\notag
\end{align}
see Lemma \ref{lemma:exp_mart}.  

We have
\begin{align}
\tilde\gamma M_{m,t}+M_{m,t}\tilde\gamma^T=&\int_0^t \frac{d}{ds}\left[ e^{-\tilde\gamma (t-s)/m}\Sigma e^{-\tilde\gamma^T (t-s)/m}\right]ds\\
=&\Sigma -e^{-\tilde\gamma t/m}\Sigma e^{- \tilde\gamma^T t/m},\notag
\end{align}
i.e. $M_{m,t}$ satisfies a Lyapunov equation. $\tilde \gamma$ has eigenvalues with positive real parts, so the  Lyapunov equation has a unique solution given by 
\begin{align}
 M_{m,t}=&\int_0^\infty e^{-\tilde\gamma \zeta}\left(\Sigma -e^{-\tilde\gamma t/m}\Sigma e^{- \tilde\gamma^T t/m}\right)e^{-\tilde\gamma^T \zeta}d\zeta.
\end{align}
See, for example, Theorem 6.4.2 in \cite{ortega2013matrix}.

We can bound the $m$-dependent part as follows:
\begin{align}\label{M_limit}
&\|\int_0^\infty e^{-\tilde\gamma \zeta}e^{-\tilde\gamma t/m}\Sigma e^{- \tilde\gamma^T t/m}e^{-\tilde\gamma^T \zeta}d\zeta\|\\
\leq & e^{-2\lambda t/m}\|\Sigma\| \int_0^\infty e^{-2\lambda \zeta}d\zeta=  \frac{\|\Sigma\|}{2\lambda}e^{-2\lambda t/m}.\notag
\end{align}
In particular, it converges to zero as $m\rightarrow 0$ and for any $t>0$
\begin{align}\label{M_limit2}
\lim_{m\rightarrow 0} M_{m,t}=\int_0^\infty  e^{-\tilde\gamma \zeta}\Sigma e^{-\tilde\gamma^T \zeta}d\zeta \equiv M
\end{align}
 where $M $ satisfies the Lyapunov equation
\begin{align}
\tilde\gamma M +M \tilde\gamma^T=\Sigma.
\end{align}

$M $ is positive-definite with
\begin{align}
k\cdot M  k\geq& \mu\int_0^\infty \|e^{-\tilde\gamma^T\zeta}k\|^2d\zeta \geq \mu \|k\|^2\int_0^\infty e^{-2\|\gamma\| \zeta}d\zeta\\
=&\frac{\mu }{2\|\gamma\|}\|k\|^2.\notag
\end{align}
Here we used the fact that the symmetric part of $\tilde \gamma^T$ has eigenvalues bounded above by $\|\gamma\|$. Therefore the eigenvalues of $M $ are bounded below by $\mu/ (2\|\gamma\|)$ and
\begin{align}\label{M_bounds}
\frac{\mu}{2\|\gamma\|}\leq \|M \|\leq \frac{\|\Sigma\|}{2\lambda}.
\end{align}
\begin{comment}
\begin{align}
 e^{-\|\gamma\| \zeta}\|k\|\leq \|e^{-\tilde\gamma\zeta}k\|\leq e^{-\lambda \zeta}\|k\|
\end{align}
\end{comment}

We are now in a position to prove the following:
\begin{lemma}\label{state_ind_lemma}
Define
\begin{align}
\tilde H=\frac{1}{(2\pi)^{n/2} |\det M |^{1/2}} \int \tilde h(z) e^{-z\cdot M ^{-1}z/2}dz,
\end{align}
where $\tilde h$ is a continuous function satisfying $|\tilde h(z)|\leq \tilde K(1+\|z\|^p)$ for some $p>0$, $\tilde K>0$. Then for any $t>0$, $q>0$ there exists $C>0$, $m_0>0$ such that whenever  $0<m\leq m_0$ we have
\begin{align}
|\tilde H_t^m-\tilde H|\leq C \tilde Ke^{-2\lambda t/m}.
\end{align}
 $C$ and $m_0$ depend only on $t$, $q$, $n$, $p$, $\mu$, $\lambda$, and upper bounds on $\|\Sigma\|$ and  $\|\gamma\|$.
\end{lemma}
\begin{proof}
Let us fix $t > 0$. To estimate
\begin{align}
\int\tilde{h}(z)\left[(\det M)^{-{1 \over 2}}\exp\left(-{1 \over 2}z\cdot M^{-1}z\right) - (\det M_{m,t})^{-{1 \over 2}}\exp\left(-{1 \over 2}z\cdot M_{m,t}^{-1}z\right)\right]\,dz
\end{align}
let us rewrite the expression in square brackets as
\begin{align}
AB - CD = A\left(B - D\right) + \left(A - C\right)D.
\end{align}
Since the eigenvalues of $M$ are bounded  above by ${\|\Sigma\|\over 2\lambda}$ and below by ${\mu \over 2\|\gamma\|}$, we have $\|M\|\leq {\|\Sigma\|\over 2\lambda}$ and $\|M^{-1}\| \le {2 \|\gamma\| \over \mu}$.  Thus, by \req{M_limit2}, for sufficiently small $m_0$, depending only on the constants listed in the statement of the Lemma, for all  $0<m\leq m_0$:
\begin{align}\label{M_mt_bounds}
\|M_{m,t}\|\leq{\|\Sigma\|\over\lambda}, \hspace{2mm}\|M_{m,t}^{-1}\| \le {3 \|\gamma\| \over \mu}.
\end{align}
  It follows that
\begin{align}
&\|M^{-1}_{m,t} - M^{-1}\| = \|M^{-1}\left(M - M_{m,t}\right)M^{-1}_{m,t}\|\\
 \leq& {6 \|\gamma\|^2 \over \mu^2}\|M - M_{m,t}\| \le {6 \|\gamma\|^2 \over \mu^2}{\|\Sigma\| \over 2\lambda}e^{-{2\lambda t \over m}}\notag
\end{align}
where in the last step we used \req{M_limit}.  We now estimate $B - D$ using the Mean Value Theorem, together with the lower bound on the quadratic forms $z\cdot M^{-1}z$ and  $z\cdot M_{m,t}^{-1}z$, which follows from \req{M_bounds} and \req{M_mt_bounds}:
\begin{align}
&\left|\exp\left(-{1 \over 2}z\cdot M^{-1}z\right) - \exp\left(-{1 \over 2}z\cdot M_{m,t}^{-1}z\right)\right|\\
 \leq& {1 \over 2}\|M^{-1} - M^{-1}_{m,t}\| \|z\|^2 \exp\left(-{\lambda \over 2\|\Sigma\|}\|z\|^2\right)\notag
\end{align}
obtaining
\begin{align}&(\det M)^{-{1 \over 2}}\left|\int \tilde{h}(z) \left(\exp\left(-{1 \over 2}z\cdot M^{-1}z\right) - \exp\left(-{1 \over 2}z\cdot M_{m,t}^{-1}z\right)\right)\right|\\
\leq&\left({\mu \over 2\|\gamma\|}\right)^{-{n \over 2}}{3 \|\gamma\|^2 \|\Sigma\| \over 2 \mu^2\lambda}\int \tilde K\left(1 + \|z\|^p\right)\|z\|^2\exp\left(-{1 \over 2}{\lambda \over \|\Sigma\|}\|z\|^2\right)\,dz \cdot e^{-{2\lambda t \over m}}.\notag
\end{align}
On the right-hand side all the factors preceding $e^{-{2\lambda t \over m}}$ are constants which depend only on the parameters listed in the statement of the Lemma.

To estimate the term containing $\left(A - C\right)D$ in the original integral, we first write
\begin{align}
&\left|\left(\det M\right)^{-{1 \over 2}} - \left(\det M_{m,t}\right)^{-{1 \over 2}}\right| = {\left|\left(\det M\right)^{{1 \over 2}} - \left(\det M_{m,t}\right)^{{1 \over 2}}\right| \over \left|\det M\right|^{1 \over 2}\left|\det M_{m,t}\right|^{1 \over 2}} \\
=&{\left|\det M - \det M_{m,t}\right| \over \left|\det M\right|^{1 \over 2}\left|\det M_{m,t}\right|^{1 \over 2}\left(\left(\det M\right)^{{1 \over 2}} + \left(\det M_{m,t}\right)^{{1 \over 2}}\right)}.\notag
\end{align}
The denominator is bounded below by $2\left({\mu \over 3\|\gamma\|}\right)^{3n/2}$. Denoting the eigenvalues of $M$ by $e_1 \le \dots \le e_n$ and those of $M_{m,t}$ by $f_1 \le \dots \le f_n$, we have by the Minimum-Maximum Principle
\begin{align}
|e_j - f_j| \le \|M - M_{m,t}\| \le {\|\Sigma\| \over 2\lambda}e^{-{2\lambda t \over m}}.
\end{align}
The difference of determinants
\begin{align}
\det M - \det M_{m,t} = e_1 \dots e_n - f_1 \dots f_n
\end{align}
can be expanded in a telescopic sum
\begin{align}
&[e_1e_2\dots e_n - f_1 e_2 \dots e_n] + [f_1e_2 \dots e_n - f_1f_2e_3\dots e_n] + \dots  \\
&+ [f_1\dots f_{j-1}e_j\dots e_n - f_1\dots f_{j-1}f_j e_{j+1}\dots e_n] \notag\\
&+[f_1f_2\dots f_{n-1}e_n - f_1f_2\dots f_{n-1}f_n]\notag
\end{align}
in which each of the $n$ terms is bounded in absolute value by $\frac{1}{2}\left(\frac{\|\Sigma\|}{ \lambda}\right)^ne^{-\frac{2\lambda t}{ m}}$, so that
\begin{align}
\int\left|\tilde{h}(z)\right|\left|\left(\det M\right)^{-{1 \over 2}} - \left(\det M_{m,t}\right)^{-{1 \over 2}}\right| \exp\left(-{1 \over 2}z\cdot M_{m,t}^{-1}z\right)\,dz
\end{align}
is bounded above by
\begin{align}
\frac{n}{4}\left(\|\Sigma\| \over \lambda\right)^n\left(3\|\gamma\| \over \mu\right)^{3n/2}\int \tilde K\left(1 + \|z\|^p\right)\exp\left(-{1 \over 2}{\lambda \over \|\Sigma\|}\|z\|^2\right)\,dz \cdot e^{-{2\lambda t \over m}}
\end{align}
where again all factors preceding $e^{-{2\lambda t \over m}}$ are constants which depend only on the parameters listed in the statement of the Lemma. 
\end{proof}

This completes the proof that, in the time and state-independent case with zero forcing, the rescaled momentum converges to a Gaussian in the small mass limit.  While illustrative, this simplified case leaves out a large part of the full derivation;  here we did not have to consider the interplay of the position and scaled momentum processes. In the next section we show how, via a sequence of simplifications, we can reduce the general case to the one treated here.

\subsection{A Sequence of Simplifications}\label{sec:simp_seq}

To reduce the general case to the one considered in the previous section, we will derive a sequence of approximations, $z_{j,t}^m$, to the scaled velocity, $z_{0,t}^m\equiv z_t^m$.  These will be approximations in the sense that $E[\|z_{j-1,t}^m-z_{j,t}^m\|^{p_j}]^{1/{p_j}}\to 0$ as $m\to 0$ for some $p_j\geq 1$.

  We will call these processes `simplifications' or `reductions' of $z_t^m$, the idea being that, for the purpose of computing the limiting joint distribution, they can be used in place of $z_t^m$. 

The end result of these (seven) reductions will be the processes
\begin{align}\label{simp7_preview}
z^m_{7,t}\equiv \frac{1}{{m}^{3/2}}\int_{t-m^\kappa}^te^{-\tilde\gamma(t,q_{t-m^\kappa})(t-s)/m} \tilde \gamma(t,q_{t-m^\kappa}) \sigma(t,q_{t-m^\kappa}) (W_t-W_s) ds,
\end{align}
where $\kappa\in (0,1)$. Note that $z_{7,t}^m$ only depends on the $q$ and $W$ processes through $q_{t-m^\kappa}$ and $W_t-W_s$ for $s\in[t-m^\kappa,t]$. These processes are independent, and so we will have effectively reduced the problem to the time and state-independent case, allowing us to use Lemma \ref{state_ind_lemma}.

The intuitive aim behind each simplified process we define below is to show that, in the manner described above, for small $m$ the processes $z_t^m$ are `essentially'  determined by only the  current position and an independent Wiener process. In this light, it is not surprising that the  limiting joint distribution is Gaussian in the scaled momentum variables.
  
{\bf Simplification 1: Reduction to}\\
\begin{align}\label{simp1}
z^m_{1,t}\equiv \frac{1}{\sqrt{m}} \Phi^m_{t} \int_0^{t}(\Phi^m_s)^{-1} \sigma(s,q^m_s) dW_s.
\end{align}
As our first simplification, we show that the initial condition and forcing terms do not contribute in the limit $m\to 0$.

To begin, let $p\geq 1$ and use  \req{fund_sol_decay}, \req{u_sol}, \req{z_def}, and Assumption \ref{assump:bounds} to compute 
\begin{align}
&E\left[\left\|z^m_t-\frac{1}{\sqrt{m}} \Phi^m_{t} \int_0^{t}(\Phi^m_s)^{-1} \sigma(s,q^m_s) dW_s\right\|^p\right]^{1/p}\\
=&E\left[\left\|\frac{1}{\sqrt{m}}\Phi^m_t\left(u^m_0+\int_0^t (\Phi^m_s)^{-1}F(s,q^m_s) ds\right)\right\|^p\right]^{1/p}\notag\\
\leq &E\left[ \frac{1}{\sqrt{m}}\left(\|\Phi^m_t\| \|u^m_0\|+\int_0^t \|\Phi^m_t(\Phi^m_s)^{-1}\| \|F(s,q^m_s)\| ds\right)\right]\notag\\
\leq &\frac{1}{\sqrt{m}}\left(e^{-\lambda t/m}E[\|u^m_0\|]+\|F\|_\infty\int_0^t e^{-\lambda (t-s)/m}ds\right)\notag\\
\leq &C^{1/2}e^{-\lambda t/m}+\frac{\sqrt{m}\|F\|_\infty}{  \lambda} (1-e^{-\lambda t/m})  =O(m^{1/2})\notag
\end{align}
as $m\rightarrow 0$.  Therefore  $E[\|z_t^m-z_{1,t}^m\|^p]^{1/p}=O(m^{1/2})$.  Hence, the forcing, $F$, and initial condition, $u_0^m$, plays no role in the limiting distribution.

{\bf Simplification 2: Reduction to}
\begin{align}\label{simp2}
z_{2,t}^m\equiv \frac{1}{m^{3/2}}\Phi^m_t\int_0^t(\Phi^m_s)^{-1} \tilde \gamma(s,q^m_s) \int_s^t \sigma(r,q^m_r) dW_rds.
\end{align}
Next we use the stochastic convolution decomposition \req{convol_decomp} to rewrite $z_{1,t}^m$, \req{simp1}, and show that the first term in that formula is also negligible in the limit.  To that end, we have defined $z_{2,t}^m$ to consist of the second term in the decomposition.

Utilizing the Burkholder-Davis-Gundy inequality  (see, for example, Theorem 3.28 in \cite{karatzas2014brownian}), for $p>1$ we obtain
\begin{align}
 &E\left[\left\|\frac{1}{\sqrt{m}}\Phi^m_t \int_0^t \sigma(s,q^m_s) dW_s\right\|^p\right]^{1/p}\\
\leq &\frac{1}{\sqrt{m}} e^{-\lambda t/m}E\left[\left \|\int_0^t \sigma(s,q^m_s) dW_s\right\|^p\right]^{1/p}\notag\\
\leq &\frac{\tilde C}{\sqrt{m}} e^{-\lambda t/m}E\left[ \left(\int_0^t \|\sigma(s,q^m_s)\|^2 ds\right)^{p/2}\right]^{1/p}\notag\\
\leq &\frac{\tilde C}{\sqrt{m}} e^{-\lambda t/m}\|\sigma\|_\infty t^{1/2}=O(m^{1/2})\notag
\end{align}
as $m\rightarrow 0$. The constant $\tilde C$, independent of $m$, comes from the use of the Burkholder-Davis-Gundy inequality.  Therefore $E[\|z_{1,t}^m-z_{2,t}^m\|^p]^{1/p}=O(m^{1/2})$ as desired.

{\bf Simplification 3: Reduction to}\\
\begin{align}\label{simp3}
z^m_{3,t}\equiv \frac{1}{m^{3/2}}\Phi^m_t\int_0^t(\Phi^m_s)^{-1} \tilde \gamma(s,q_s) \int_s^t \sigma(r,q_r) dW_r.
\end{align}
We now show that the process $q_t^m$ in $z_{2,t}^m$, \req{simp2}, can be replaced by $q_t$, the solution to the limiting SDE, \req{q_SDE}.

 The following computation uses the Burkholder-Davis-Gundy inequality, Minkowski's inequality for integrals, and H\"older's inequality to show that we can  replace $q_t^m$ with $q_t$ in \req{simp2}.    In the following, $p\geq 2$ and $\tilde C$ is a constant that potentially varies line to line.
\begin{align}
&E\left[\left\|\frac{1}{m^{3/2}}\Phi^m_t\int_0^t(\Phi^m_s)^{-1} \tilde \gamma(s,q^m_s) \int_s^t \sigma(r,q^m_r) dW_rds\right.\right.\\
&\left.\left.-\frac{1}{m^{3/2}}\Phi^m_t\int_0^t(\Phi^m_s)^{-1} \tilde \gamma(s,q_s) \int_s^t \sigma(r,q_r) dW_rds\right\|^p\right]^{1/p}\notag\\
\leq&E\left[\left\|\frac{1}{m^{3/2}}\Phi^m_t\int_0^t(\Phi^m_s)^{-1} \tilde \gamma(s,q^m_s) \int_s^t \left(\sigma(r,q^m_r) - \sigma(r,q_r)\right) dW_r ds\right\|^p\right]^{1/p}\notag\\
&+E\left[\left\|\frac{1}{m^{3/2}}\Phi^m_t\int_0^t(\Phi^m_s)^{-1} \left(\tilde \gamma(s,q^m_s) -\tilde \gamma(s,q_s)\right) \int_s^t \sigma(r,q_r) dW_r ds\right\|^p\right]^{1/p}\notag\\
\leq&\frac{1}{m^{3/2}}E\left[\left(\int_0^t e^{-\lambda(t-s)/m}\|\tilde \gamma\|_\infty \left\|\int_s^t \left(\sigma(r,q^m_r) - \sigma(r,q_r)\right) dW_r\right\| ds\right)^p\right]^{1/p}\notag\\
&+\frac{1}{m^{3/2}}E\left[\left(\int_0^t e^{-\lambda(t-s)/m} \|\tilde \gamma(s,q^m_s) -\tilde \gamma(s,q_s)\| \left\|\int_s^t \sigma(r,q_r) dW_r\right\| ds\right)^p\right]^{1/p}\notag\\
\leq&\frac{\|\tilde \gamma\|_\infty}{m^{3/2}} \int_0^t e^{-\lambda(t-s)/m} E\left[ \left\|\int_s^t \left(\sigma(r,q^m_r) - \sigma(r,q_r)\right) dW_r\right\|^p \right]^{1/p}ds\notag\\
&+\frac{1}{m^{3/2}}\int_0^t e^{-\lambda(t-s)/m} E\left[\|\tilde \gamma(s,q^m_s) -\tilde \gamma(s,q_s)\|^p \left\|\int_s^t \sigma(r,q_r) dW_r\right\|^p\right]^{1/p} ds\notag\\
\leq&\frac{\|\tilde \gamma\|_\infty\tilde C}{m^{3/2}} \int_0^t e^{-\lambda(t-s)/m} E\left[ \left(\int_s^t\|q^m_r-q_r\|^2 dr\right)^{p/2} \right]^{1/p}ds\notag\\
&+\frac{\tilde C}{m^{3/2}}\int_0^t e^{-\lambda(t-s)/m} E\left[\|q^m_s -q_s\|^{2p}\right]^{1/(2p)}E\left[ \left\|\int_s^t \sigma(r,q_r) dW_r\right\|^{2p}\right]^{1/(2p)} \!\!\!\!\!ds\notag\\
\leq&\frac{\|\tilde \gamma\|_\infty\tilde C}{m^{3/2}} \int_0^t e^{-\lambda(t-s)/m}  \left(\int_s^tE\left[ \|q^m_r-q_r\|^p\right]^{2/p}dr\right)^{1/2}ds\notag\\
&+\frac{\tilde C \|\sigma\|_\infty}{m^{3/2}}\int_0^t e^{-\lambda(t-s)/m}(t-s)^{1/2} ds\sup_{0\leq s\leq t} E\left[\|q^m_s -q_s\|^{2p}\right]^{1/(2p)}\notag\\
\leq&\frac{\|\tilde \gamma\|_\infty\tilde C}{m^{3/2}} \int_0^t e^{-\lambda(t-s)/m}  (t-s)^{1/2}ds\sup_{0\leq r\leq t}E\left[ \|q^m_r-q_r\|^p\right]^{1/p}\notag\\
&+\frac{\tilde C \|\sigma\|_\infty}{m^{3/2}}\int_0^t e^{-\lambda(t-s)/m}(t-s)^{1/2} ds\sup_{0\leq s\leq t} E\left[\|q^m_s -q_s\|^{2p}\right]^{1/(2p)}\notag
\end{align}
\begin{align}
\leq&\tilde C \int_0^{t/m} e^{-\lambda r}  r^{1/2}dr\left(\sup_{0\leq r\leq t}E\left[ \|q^m_r-q_r\|^p\right]^{1/p}+ \sup_{0\leq r\leq t}E\left[\|q^m_r -q_r\|^{2p}\right]^{1/(2p)}\right)   \notag\\
=&O(m^{1/2})\notag
\end{align}
as $m\rightarrow 0$, by \req{results_summary1}.  Therefore $E[\|z_{2,t}^m-z_{3,t}^m\|^p]^{1/p}=O(m^{1/2})$ as desired.

{\bf Simplification 4: Reduction to}\\
\begin{align}\label{simp4}
z^m_{4,t}\equiv \frac{1}{m^{3/2}}\Phi^m_t\int_0^t(\Phi^m_s)^{-1} \tilde \gamma(s,q_s)\sigma(s,q_s) (W_t-W_s) ds.
\end{align}
Next we show that, due to the fast decay of  $\Phi_t^m(\Phi_s^m)^{-1}$ for $s<t$, the inner (stochastic) integral, $\int_s^t \sigma(r,q^m_r) dW_r$, in $z^m_{3,t}$, \req{simp3}, can be approximated by $\sigma(s,q_s) (W_t-W_s)$.

First, compute the bound
\begin{align}
&\left\|\frac{1}{m^{3/2}}\Phi^m_t\int_0^t(\Phi^m_s)^{-1} \tilde \gamma(s,q_s)\int_s^t \sigma(r,q_r) dW_r ds\right.\\
&\left.-\frac{1}{m^{3/2}} \Phi^m_t\int_0^t(\Phi^m_s)^{-1} \tilde \gamma(s,q_s) \sigma(s,q_s) (W_t-W_s) ds\right\|\notag\\
=&\left\|\frac{1}{m^{3/2}}\Phi^m_t\int_0^t(\Phi^m_s)^{-1} \tilde \gamma(s,q_s) \int_s^t \left(\sigma(r,q_r) -\sigma(s,q_s)\right)  dW_r ds\right\|\notag\\
\leq & \frac{1}{m^{3/2}} \|\tilde \gamma\|_\infty \int_0^t  e^{-\lambda (t-s)/m} \left\|\int_s^t \left(\sigma(r,q_r) -\sigma(s,q_s)\right)  dW_r\right\| ds.\notag
\end{align}
Then, using Burkholder-Davis-Gundy inequality and the Minkowski inequality for integrals, for $p\geq 2$ we compute a bound on the $L^p$-norm 
\begin{align}
&E\left[\left\|\frac{1}{m^{3/2}}\Phi^m_t\int_0^t(\Phi^m_s)^{-1} \tilde \gamma(s,q_s)\int_s^t \sigma(r,q_r) dW_r ds\right.\right.\\
&\left.\left.- \frac{1}{m^{3/2}}\Phi^m_t\int_0^t(\Phi^m_s)^{-1} \tilde \gamma(s,q_s) \sigma(s,q_s) (W_t-W_s) ds\right\|^p\right]^{1/p}\notag\\
\leq & \frac{1}{m^{3/2}} \|\tilde \gamma\|_\infty E\left[ \left(\int_0^t e^{-\lambda (t-s)/m} \left\|\int_s^t \left(\sigma(r,q_r) -\sigma(s,q_s) \right) dW_r\right\| ds\right)^p\right]^{1/p}\notag\\
\leq & \frac{1}{m^{3/2}}\|\tilde \gamma\|_\infty   \int_0^tE\left[ \left(e^{-\lambda (t-s)/m}\left\|\int_s^t \left(\sigma(r,q_r) -\sigma(s,q_s) \right) dW_r\right\| \right)^p\right]^{1/p}ds\notag\\
\leq &\frac{1}{m^{3/2}}\tilde C   \int_0^t  e^{-\lambda (t-s)/m} E\left[  \left(\int_s^t \|\sigma(r,q_r) -\sigma(s,q_s)\|^2 dr\right)^{p/2}\right]^{1/p}ds\notag
\end{align}
\begin{align}
\leq &\frac{1}{m^{3/2}}\tilde C   \int_0^t  e^{-\lambda (t-s)/m} E\left[  \left(\int_s^t (r-s)^2+ \|q_r -q_s\|^2 dr\right)^{p/2}\right]^{1/p}ds\notag\\
\leq & \frac{1}{m^{3/2}}\tilde C   \int_0^t  e^{-\lambda (t-s)/m}  \left((t-s)^{3/2}+  E\left[ \left(\int_s^t\|q_r -q_s\|^2 dr\right)^{p/2}\right]^{1/p}\right)ds\notag\\
\leq & \frac{1}{m^{3/2}}\tilde C   \int_0^t  e^{-\lambda (t-s)/m}\left((t-s)^{3/2}+\left( \int_s^t E\left[  \|q_r -q_s\|^p\right]^{2/p} dr \right)^{1/2}\right)ds\notag\\
\leq & \frac{1}{m^{3/2}}\tilde C   \int_0^t  e^{-\lambda (t-s)/m}\left((t-s)^{3/2}+ (t-s)^{1/2}\sup_{s\leq r\leq t} E\left[  \|q_r -q_s\|^p\right]^{1/p} \right) ds\notag\\
= & \tilde C   \int_0^{t/m}  e^{-\lambda u} u^{1/2} \left(mu+\sup_{t-mu\leq r\leq t}E\left[  \|q_r -q_{t-mu}\|^p\right]^{1/p}\right)  du\notag\\
= & O(m)+ \tilde C   \int_0^{t/m}  e^{-\lambda u} u^{1/2} \sup_{t-mu\leq r\leq t}E\left[  \|q_r -q_{t-mu}\|^p\right]^{1/p}  du.\notag
\end{align}

 By Lemma \ref{q_cont_lemma}, for any $0\leq s\leq t$ we have
\begin{align}
\sup_{s\leq r\leq t}E\left[  \|q_r -q_s\|^p\right]^{1/p}\leq\tilde C\left((t-s)+\left(t-s \right)^{1/2}\right)
\end{align}
where $\tilde C$ depends only on $p$, $n$, and the drift vector field and diffusion matrix of the SDE for $q_t$.  Therefore
\begin{align}
&E\left[\left\|\frac{1}{{m}^{3/2}}\Phi^m_t\int_0^t(\Phi^m_s)^{-1} \tilde \gamma(s,q_s)\int_s^t \sigma(r,q_r) dW_r ds\right.\right.\\
&\left.\left.-\frac{1}{{m}^{3/2}} \Phi^m_t\int_0^t(\Phi^m_s)^{-1} \tilde \gamma(s,q_s) \sigma(s,q_s) (W_t-W_s) ds\right\|^p\right]^{1/p}\notag\\
\leq &O(m)+ \tilde C   \int_0^{t/m}  e^{-\lambda u} u^{1/2} \tilde C\left(mu+m^{1/2} u^{1/2}\right) du =O(m^{1/2}).\notag
\end{align}
This proves that  $E[\|z_{3,t}^m-z_{4,t}^m\|^p]^{1/p}=O(m^{1/2})$ as desired.

{\bf Simplification 5: Reduction to}\\
\begin{align}\label{simp5}
z^m_{5,t}\equiv \frac{1}{m^{3/2}}\Phi^m_t\int_{t-m^\kappa}^t(\Phi^m_s)^{-1} \tilde \gamma(s,q_s)\sigma(s,q_s) (W_t-W_s) ds, \hspace{2mm} \kappa\in(0,1).
\end{align}
The fast decay of $\Phi_t^m(\Phi_s^m)^{-1}$ for $s<t$ makes the majority of the integral from $0$ to $t$ in $z_{4,t}^m$, \req{simp4}, negligible.  More precisely, for $p>1$,  $0<\kappa<1$ we have
\begin{align}
&E\left[\left\|\frac{1}{m^{3/2}} \Phi^m_t\int_0^t(\Phi^m_s)^{-1} \tilde \gamma(s,q_s) \sigma(s,q_s) (W_t-W_s) ds\right.\right.\\
&\left.\left.-\frac{1}{m^{3/2}}\Phi^m_t\int_{t-m^\kappa}^t(\Phi^m_s)^{-1} \tilde \gamma(s,q_s) \sigma(s,q_s) (W_t-W_s) ds  \right\|^p\right]^{1/p}\notag\\
\leq &E\left[\left(\frac{1}{m^{3/2}} \|\tilde \gamma\|_\infty\|\sigma\|_\infty\int_0^{t-m^\kappa}e^{-\lambda (t-s)/m} \|W_t-W_s\| ds  \right)^p\right]^{1/p}\notag\\
\leq &E\left[\left(\frac{1}{m^{3/2}} \|\tilde \gamma\|_\infty\|\sigma\|_\infty e^{-\lambda m^{\kappa-1}}t \sup_{0\leq s\leq t}\|W_t-W_s\|  \right)^p\right]^{1/p}\notag\\
\leq &\frac{2}{m^{3/2}} \|\tilde \gamma\|_\infty\|\sigma\|_\infty e^{-\lambda m^{\kappa-1}}tE\left[ \sup_{0\leq s\leq t}\|W_s\|  ^p\right]^{1/p}=O(m^q)\notag
\end{align}
as $m\rightarrow0$ for any $q>0$. In particular, $E[\|z_{4,t}^m-z_{5,t}^m\|^p]^{1/p}=O(m^{1/2})$ as desired.

{\bf Simplifiction 6: Reduction to}\\
\begin{align}\label{simp6}
z_{6,t}^m\equiv \frac{1}{{m}^{3/2}}\Phi^m_t\int_{t-m^\kappa}^t(\Phi^m_s)^{-1} \tilde \gamma(t,q_{t-m^\kappa}) \sigma(t,q_{t-m^\kappa}) (W_t-W_s) ds.
\end{align}
(H\"older) continuity of $q_t$, both pathwise and in an $L^p$ sense, allows us to replace $q_s$ with $q_{t-m^\kappa}$ in $z_{5,t}^m$, \req{simp5}. 

 Using  Lemma \ref{q_cont_lemma} along with the Minkowski inequality for integrals and H\"older's inequality, for $p>1$ we obtain
\begin{align}
&E\left[\left\|\frac{1}{m^{3/2}}\Phi^m_t\int_{t-m^\kappa}^t(\Phi^m_s)^{-1} \tilde \gamma(s,q_s) \sigma(s,q_s) (W_t-W_s) ds  \right.\right.\\
&\left.\left.-\frac{1}{m^{3/2}}\Phi^m_t\int_{t-m^\kappa}^t(\Phi^m_s)^{-1} \tilde \gamma(t,q_{t-m^\kappa}) \sigma(t,q_{t-m^\kappa}) (W_t-W_s) ds  \right\|^p\right]^{1/p}\notag\\
\leq &\frac{\tilde C}{m^{3/2}}E\left[\left(\int_{t-m^\kappa}^te^{-\lambda(t-s)/m} (t-s+\|q_s-q_{t-m^\kappa}\|) \|W_t-W_s\| ds   \right)^p\right]^{1/p}\notag\\
\leq &\frac{\tilde C}{m^{3/2}}\int_{t-m^\kappa}^te^{-\lambda(t-s)/m}E\left[ (t-s+\|q_s-q_{t-m^\kappa}\|)^p \|W_t-W_s\|^p\right]^{1/p} ds   \notag\\
\leq &\frac{\tilde C}{m^{3/2}}\int_{t-m^\kappa}^te^{-\lambda(t-s)/m}E\left[ (t-s+\|q_s-q_{t-m^\kappa}\|)^{2p}\right]^{1/(2p)}E\left[ \|W_t-W_s\|^{2p}\right]^{1/(2p)} \!\!ds.   \notag
\end{align}

\begin{comment}
We compute
\begin{align}
&E\left[\|W_t-W_s\|^{2p}\right]=\frac{1}{(2\pi(t-s))^{n/2}}\int \|z\|^{2p}e^{-\|z\|^2/(2(t-s))}dz\\
=&\frac{1}{(2\pi(t-s))^{n/2}}(t-s)^{p}\int \|u\|^{2p}e^{-\|u\|^2/2}(t-s)^{n/2}du\notag\\
=&\tilde C (t-s)^{p},\notag
\end{align}
\end{comment}

The rightmost expected value can be computed 
\begin{align}
&E\left[\|W_t-W_s\|^{2p}\right]=\frac{(t-s)^{p}}{(2\pi)^{n/2}}\int \|u\|^{2p}e^{-\|u\|^2/2}du
\end{align}
and therefore
\begin{align}
&E\left[\left\|\frac{1}{m^{3/2}}\Phi^m_t\int_{t-m^\kappa}^t(\Phi^m_s)^{-1} \tilde \gamma(s,q_s) \sigma(s,q_s) (W_t-W_s) ds  \right.\right.\\
&\left.\left.-\frac{1}{m^{3/2}}\Phi^m_t\int_{t-m^\kappa}^t(\Phi^m_s)^{-1} \tilde \gamma(t,q_{t-m^\kappa}) \sigma(t,q_{t-m^\kappa}) (W_t-W_s) ds  \right\|^p\right]^{1/p}\notag\\
\leq &\frac{\tilde C}{m^{3/2}}\int_{t-m^\kappa}^te^{-\lambda(t-s)/m}(t-s)^{1/2}\left (t-s+E\left[\|q_s-q_{t-m^\kappa}\|^{2p}\right]^{1/(2p)}\right) ds \notag\\
\leq &\frac{\tilde C}{m^{3/2}}\int_{t-m^\kappa}^te^{-\lambda(t-s)/m}(t-s)^{1/2}\left (t-s+m^\kappa+m^{\kappa/2}\right) ds \notag\\
= &\tilde C\int_0^{m^{-(1-\kappa)}}e^{-\lambda r} r^{1/2}\left (mr+m^\kappa+m^{\kappa/2}\right) dr=O(m^{\kappa/2}). \notag
\end{align}
Hence $E[\|z_{5,t}^m-z_{6,t}^m\|^p]^{1/p}=O(m^{\kappa/2})$ as desired. Note that the convergence rate bound has weakened from $O(m^{1/2})$ to $O(m^{\kappa/2})$ with $\kappa\in (0,1)$.

{\bf Simplifiction 7: Reduction to}\\
\begin{align}\label{simp7}
z^m_{7,t}\equiv \frac{1}{{m}^{3/2}}\int_{t-m^\kappa}^te^{-\tilde\gamma(t,q_{t-m^\kappa})(t-s)/m} \tilde \gamma(t,q_{t-m^\kappa}) \sigma(t,q_{t-m^\kappa}) (W_t-W_s) ds.
\end{align}
The fast decay of $\Phi_t^m(\Phi_s^m)^{-1}$ as $s$ gets far from $t$ also means that we will be able to replace $\Phi_t^m(\Phi_s^m)^{-1}$ with $e^{-\tilde\gamma(t,q_{t-m^\kappa})(t-s)/m}$ in $z_{6,t}^m$, \req{simp6}.

 Using Lemma \ref{fund_matrix_diverg} along with the Minkowski inequality for integrals and H\"older's inequality, for $p>1$ we obtain the bounds
\begin{align}
&E\left[\left\|\frac{1}{m^{3/2}}\Phi^m_t\int_{t-m^\kappa}^t(\Phi^m_s)^{-1} \tilde \gamma(t,q_{t-m^\kappa}) \sigma(t,q_{t-m^\kappa}) (W_t-W_s) ds  \right.\right.\\
&\left.\left.-\frac{1}{m^{3/2}}\int_{t-m^\kappa}^te^{-\tilde\gamma(t,q_{t-m^\kappa})(t-s)/m} \tilde \gamma(t,q_{t-m^\kappa}) \sigma(t,q_{t-m^\kappa}) (W_t-W_s) ds  \right\|^p\right]^{1/p}\notag\\
\leq &\frac{\|\tilde \gamma  \sigma\|_\infty}{m^{3/2}}E\left[\left(\int_{t-m^\kappa}^t\left\|\Phi^m_t(\Phi^m_s)^{-1}-e^{-\tilde\gamma(t,q_{t-m^\kappa})(t-s)/m}\right\|  \left\|W_t-W_s\right\| ds    \right)^p\right]^{1/p}\notag\\
\leq &\frac{\|\tilde \gamma  \sigma\|_\infty}{m^{3/2}}E\left[\left(\int_{t-m^\kappa}^t\int_s^t\| \tilde\gamma(r,q_{r}^m)/m-\tilde\gamma(t,q_{t-m^\kappa})/m\|dre^{-\lambda (t-s)/m} \left\|W_t-W_s\right\| ds    \right)^p\right]^{1/p}\notag\\
\leq &\frac{\tilde C\|\tilde \gamma  \sigma\|_\infty}{m^{5/2}}E\left[\left(\int_{t-m^\kappa}^t\int_s^t\left( t-r +\|q_{r}^m-q_{t-m^\kappa}\|\right)dre^{-\lambda (t-s)/m} \left\|W_t-W_s\right\| ds    \right)^p\right]^{1/p}\notag\\
\leq &\frac{\tilde C\|\tilde \gamma  \sigma\|_\infty}{m^{5/2}} \int_{t-m^\kappa}^t e^{-\lambda (t-s)/m}\int_s^tE\left[\left( t-r +\|q_{r}^m-q_{t-m^\kappa}\|\right)^p \left\|W_t-W_s\right\|^p\right]^{1/p} drds  \notag\\
\leq &\frac{\tilde C\|\tilde \gamma  \sigma\|_\infty}{m^{5/2}} \int_{t-m^\kappa}^t \bigg(e^{-\lambda (t-s)/m} E\left[ \left\|W_t-W_s\right\|^{2p}\right]^{1/(2p)}\notag\\ 
&\hspace{2.5cm}\times \int_s^tE\left[\left( t-r +\|q_{r}^m-q_{t-m^\kappa}\|\right)^{2p}\right]^{1/(2p)}dr\bigg)ds.  \notag
\end{align}

Therefore, using Lemma \ref{q_cont_lemma}, 
\begin{align}
&E\left[\left\|\frac{1}{m^{3/2}}\Phi^m_t\int_{t-m^\kappa}^t(\Phi^m_s)^{-1} \tilde \gamma(t,q_{t-m^\kappa}) \sigma(t,q_{t-m^\kappa}) (W_t-W_s) ds  \right.\right.\\
&\left.\left.-\frac{1}{m^{3/2}}\int_{t-m^\kappa}^te^{-\tilde\gamma(t,q_{t-m^\kappa})(t-s)/m} \tilde \gamma(t,q_{t-m^\kappa}) \sigma(t,q_{t-m^\kappa}) (W_t-W_s) ds  \right\|^p\right]^{1/p}\notag\\
\leq &\frac{\tilde C}{m^{5/2}} \int_{t-m^\kappa}^t e^{-\lambda (t-s)/m}(t-s)^{1/2} \left( (t-s)^2 +\int_s^tE\left[\|q_{r}^m-q_r\|^{2p}\right]^{1/(2p)}\right.\notag\\
&\left.+E\left[\|q_{r}-q_{t-m^\kappa}\|^{2p}\right]^{1/(2p)}dr\right)ds  \notag\\
\leq &\frac{\tilde C}{m^{5/2}} \int_{t-m^\kappa}^t e^{-\lambda (t-s)/m}(t-s)^{3/2} \left( (t-s) +\sup_{s\leq r\leq t}E\left[\|q_{r}^m-q_r\|^{2p}\right]^{1/(2p)}\right.\notag\\
&\left.+\sup_{s\leq r\leq t}E\left[\|q_{r}-q_{t-m^\kappa}\|^{2p}\right]^{1/(2p)}\right)ds  \notag\\
\leq &\frac{\tilde C}{m^{5/2}} \int_{t-m^\kappa}^t e^{-\lambda (t-s)/m}(t-s)^{3/2} \left( (t-s) +\sup_{s\leq r\leq t}E\left[\|q_{r}^m-q_r\|^{2p}\right]^{1/(2p)}\right.\notag\\
&\left.+(t-(t-m^\kappa))+(t-(t-m^\kappa))^{1/2}\right)ds  \notag\\
= &\tilde C \int_0^{m^{-(1-\kappa)}} e^{-\lambda r}r^{3/2} \left( mr +m^\kappa+m^{\kappa/2}\right)dr  \notag\\
&+\tilde C\int_0^{m^{-(1-\kappa)} } e^{-\lambda r}r^{3/2} dr \sup_{s\leq r\leq t}E\left[\|q_{r}^m-q_r\|^{2p}\right]^{1/(2p)}\notag\\
=&O(m^{\kappa/2})+O(m^{1/2})=O(m^{\kappa/2}).\notag
\end{align}
Therefore $E[\|z_{6,t}^m-z_{7,t}^m\|^p]^{1/p}=O(m^{\kappa/2})$ as desired.

\subsection{Dependence of $z_{7,t}^m$ on the Processes $q$ and $W$ }
$z_{7,t}^m$, \req{simp7}, is the last in our sequence of simplified processes, and the process we will use in our computation of the small mass limit of the Fourier transform. It will be useful later on to write $z_{7,t}^m$ as
\begin{align}\label{z_G_formula}
&z^m_{7,t}= G_t^m(q_{t-m^\kappa}, W^{m,t}),
\end{align}
where we have defined the Wiener processes
\begin{align}\label{W_mt_def}
 W^{m,t}\equiv (W_{\cdot+(t-m^\kappa)}-W_{t-m^\kappa})|_{[0,1]}
\end{align}
 and the map $G_t^m:\mathbb{R}^n\times C([0,1],\mathbb{R}^n)\rightarrow \mathbb{R}^n$,
\begin{align}\label{G_mt_def}
 G_t^m(q,y)=\frac{1}{m^{3/2}} \int_{t-m^\kappa}^{t} e^{-\tilde \gamma(t,q)(t-s)/m}\tilde \gamma(t,q)\sigma(t,q)( y(m^\kappa)-y(s-(t-m^\kappa)))ds
\end{align}
(for simplicity, we take $m\leq 1$ so that $y(s)$ need only be defined for $s\in[0,1]$).  Note that $G_t^m$ is jointly continuous in $(q,y)$ and satisfies the following estimates
\begin{align}\label{G_bound}
\|G_t^m(q,y)\|\leq&\frac{\|\tilde \gamma\sigma\|_\infty}{m^{3/2}} \int_{t-m^\kappa}^t e^{-\lambda(t-s)/m} \|y(m^\kappa)-y(s-(t-m^\kappa))\|ds\\
\leq &\frac{\|\tilde \gamma\sigma\|_\infty}{\lambda m^{1/2}} \sup_{t-m^\kappa\leq s\leq t}\|y(m^\kappa)-y(s-(t-m^\kappa))\|\notag
\end{align}
and
\begin{align}\label{G_q_Lipschitz}
&\| G_t^m(q,y)-G_t^m(\tilde q,y)\|\\
\leq& \frac{1}{m^{3/2}} \int_{t-m^\kappa}^t\|e^{-\tilde \gamma(t,q)(t-s)/m}\tilde \gamma(t,q)\sigma(t,q)-e^{-\tilde \gamma(t,\tilde q)(t-s)/m}\tilde \gamma(t,\tilde q)\sigma(t,\tilde q)\|ds\notag\\
&\times \sup_{t-m^\kappa\leq s\leq t}\|y(m^\kappa)-y(s-(t-m^\kappa))\|\notag\\
\leq&\bigg( \frac{\| \gamma\sigma\|_\infty}{m^{3/2}} \int_{t-m^\kappa}^t\|e^{-\tilde \gamma(t,q)(t-s)/m} -e^{-\tilde \gamma(t,\tilde q)(t-s)/m}\|ds\notag\\
& +\frac{\tilde C}{m^{3/2}} \int_{t-m^\kappa}^te^{-\lambda(t-s)/m}\|q-\tilde q\|ds\bigg) \sup_{t-m^\kappa\leq s\leq t}\|y(m^\kappa)-y(s-(t-m^\kappa))\|\notag\\
\leq& \left(\frac{\| \gamma\sigma\|_\infty}{m^{3/2}} \int_{t-m^\kappa}^t(t-s) \|\tilde \gamma(t,\tilde q)/m-\tilde \gamma(t,q)/m\| e^{-\lambda (t-s)/m}ds+\frac{\tilde C}{\lambda m^{1/2}}\|q-\tilde q\| \right)\notag\\
&\times\sup_{t-m^\kappa\leq s\leq t}\|y(m^\kappa)-y(s-(t-m^\kappa))\|\notag\\
\leq& \left(\tilde C\| \gamma\sigma\|_\infty  \int_0^\infty r e^{-\lambda r}dr +\frac{\tilde C}{\lambda }\right)m^{-1/2}\|q-\tilde q\| \sup_{t-m^\kappa\leq s\leq t}\|y(m^\kappa)-y(s-(t-m^\kappa))\|.\notag
\end{align}

\subsection{The Limiting Distribution}
We now have the tools to compute the small mass limit of the distributions of $Y^m=(J^m,q^m_{t_1},...,q^m_{t_N},z_{t_1}^m,...,z_{t_N}^m)$ by computing the limit of the Fourier transforms.

\begin{theorem}\label{thm:conv_dis}
Let
\begin{align}
M(t,q)=\int_0^\infty  e^{-\tilde\gamma(t,q) \zeta}\Sigma(t,q) e^{-\tilde\gamma^T(t,q) \zeta}d\zeta,
\end{align}
$Q_{J,t}=(J,q_{t_1},...q_{t_N})$, $\mu_{J,t}$ be the distribution of $Q_{J,t}$,
\begin{align}\label{F_m_def}
\mathcal{F}^m(k)\equiv E[\exp(i  k\cdot Y^m)],
\end{align}
and
\begin{align}
\mathcal{F}(k)=E\bigg[&  \exp\left(ik_0\cdot Q_{J,t}-\sum_{j=1}^Nk_j\cdot M(t_j,q_{t_j})k_j/2\right)\bigg],
\end{align}
 where  $k\equiv (k_0,k_1,...,k_N)$, $k_0\in\mathbb{R}^{d+Nn}$, and $k_j\in \mathbb{R}^n$ for $j=1,...,N$.

  Then
\begin{align}
&\lim_{m\rightarrow 0}\mathcal{F}^m(k)=\mathcal{F}(k),
\end{align}
and hence the distributions of  $(J^m,q^m_{t_1},...,q^m_{t_N},z_{t_1}^m,...,z_{t_N}^m)$ converge weakly to 
\begin{align}
d\nu=&  \left(\prod_{j=1}^N\frac{1}{(2\pi)^{n/2} |\det M(t_j,q_j)|^{1/2}}\exp\left[-z_j\cdot M^{-1}(t_j,q_j)z_j/2\right]dz_j\right)\\
&\times \mu_{J,t}(dJ,dq_1,...,dq_N).\notag
\end{align}
\end{theorem}
\begin{proof}

Define
\begin{align}
\mathcal{F}^m_{l}(k)=E\bigg[&  \exp\left(ik_0\cdot Q_{J,t}+i\sum_{j=1}^l k_j\cdot z^m_{t_j}-\sum_{j=l+1}^Nk_j\cdot M(t_j,q_{t_j})k_j/2\right)\bigg],
\end{align}
for $l=0,...,N$.  In particular, $\mathcal{F}(k)=\mathcal{F}^m_{0}(k)$.  First note that
\begin{align}
&|\mathcal{F}^m(k)-\mathcal{F}^m_N(k)|\leq \|k_0\|\left(E\left[\|J-J^m\|^2\right]^{1/2}+\sum_{j=1}^N E\left[\|q_{t_j}-q_{t_j}^m\|\right]\right)=o(1)
\end{align}
as $m\to 0$.

We will now show that 
\begin{align}
\lim_{m\rightarrow 0} |\mathcal{F}^m_{l}(k)-\mathcal{F}^m_{l-1}(k)|=0
\end{align}
for each $l=1,...,N$, which will imply the desired result.

Given $l$, and using the calculations from Section \ref{sec:simp_seq} we have
\begin{align}
& |\mathcal{F}^m_{l}(k)-\mathcal{F}^m_{l-1}(k)|\\
=&\bigg|E\bigg[  \exp\left(ik_0\cdot Q_{J,t}+i\sum_{j=1}^{l-1} k_j\cdot z^m_{t_j}-\sum_{j=l+1}^Nk_j\cdot M(t_j,q_{t_j})k_j/2\right)\notag\\
&\times\left(\exp(ik_l\cdot z^m_{t_l})-\exp(-k_l\cdot M(t_l,q_{t_l})k_l/2)\right)\bigg]\bigg|\notag\\
\leq &o(1)+\left|E\bigg[ X\exp\left(i\sum_{j=1}^{l-1} k_j\cdot z^m_{t_j}\right)\left(e^{ik_l\cdot z^m_{7,t_l}}- e^{-k_l\cdot M(t_l,q_{t_l})k_l/2}\right)\bigg]\right|\notag\\
=&o(1)+\left|E\bigg[E( X|\mathcal{F}_{t_l})\exp\left(i\sum_{j=1}^{l-1} k_j\cdot z^m_{t_j}\right)\left(e^{ik_l\cdot z^m_{7,t_l}}- e^{-k_l\cdot M(t_l,q_{t_l})k_l/2}\right)\bigg]\right|,\notag
\end{align}
where we defined the bounded random variable
\begin{align}
X\equiv  \exp\left(ik_0\cdot Q_{J,t}-\sum_{j=l+1}^Nk_j\cdot M(t_j,q_{t_j})k_j/2\right).
\end{align}

Using Lemma \ref{mart_rep_lemma}, there exists $f\in L^2([0,\infty)\times \Omega,Prog,ds\times P,\mathbb{R}^n)$ such that 
\begin{align}
E(X|\mathcal{F}_t)=E(X|\mathcal{F}_0)+\int_{0}^t f_s dW_s,
\end{align}
and therefore, by the It\^o isometry,
\begin{align}
&E\left[\left|E(X|\mathcal{F}_t)-E(X|\mathcal{F}_{t-m^\kappa})\right|^2\right]\\
=& E\left[\left|\int_{t-m^\kappa}^t f_s dW_s\right|^2\right]= E\left[\int_{t-m^\kappa}^t \|f_s\|^2 ds\right].\notag
\end{align}
This converges to zero as $m\rightarrow 0$ because $f\in L^2(ds\times dP)$, therefore
\begin{align}
& |\mathcal{F}^m_{l,t}(k)-\mathcal{F}^m_{l-1,t}(k)|\\
\leq &o(1)+\left|E\bigg[\tilde X^m\left(e^{ik_l\cdot z^m_{7,t_l}}- e^{-k_l\cdot M(t_l,q_{t_l})k_l/2}\right)\bigg]\right|,\notag
\end{align}
where we defined
\begin{align}
\tilde X^m=E( X|\mathcal{F}_{t_l-m^\kappa})\exp\left[i\sum_{j=1}^{l-1} k_j\cdot z^m_{t_j}\right].
\end{align}
Note that, for $m$ sufficiently small, $\tilde X^m$ is $\mathcal{F}_{t_l-m^\kappa}$-measurable. 

Using \req{z_G_formula}  we have 
\begin{align}
&E\bigg[\tilde X^m e^{ik_l\cdot z^m_{7,t_l}}\bigg]=E\bigg[\tilde X^m e^{ik_l\cdot  G_{t_l}^m(q_{t_l-m^\kappa}, W^{m,t_l})}\bigg]\\
=&\int x e^{ik_l\cdot  G_{t_l}^m(q,y)} (\tilde X^m,q_{t_l-m^\kappa},W^{m,t_l})_*P(dx,dq,dy)\notag.
\end{align}
where $(\tilde X^m,q_{t_l-m^\kappa},W^{m,t_l})_*P$ denotes the pushforward measure i.e. the distribution of $(\tilde X^m,q_{t_l-m^\kappa},W^{m,t_l})$ under the probability measure $P$ (see Assumption \ref{assump:prob_space} for discussion of the assumptions made about the probability space).

$(\tilde X^m,q_{t_l-m^\kappa})$ is $\mathcal{F}_{t_l-m^\kappa}$-measurable and $W^{m,t_l}$ is a Wiener process that is independent of $\mathcal{F}_{t_l-m^\kappa}$. Therefore 
\begin{align}\label{independence_calc}
(\tilde X^m,q_{t_l-m^\kappa},W^{m,t_l})_*P=(\tilde X^m,q_{t_l-m^\kappa})_*P\times  W^{m,t_l}_*P\equiv (\tilde X^m,q_{t_l-m^\kappa})_*P\times\mu,
\end{align}
 where $\mu$ is the Wiener measure on the path space $C([0,1],\mathbb{R}^n)$ (note that $\mu$ does not depend on $m$).

We can now write
\begin{align}
&E\bigg[\tilde X^m e^{ik_l\cdot z^m_{7,t_l}}\bigg]=E\left[ \tilde X^m \int  e^{ik_l\cdot  G_{t_l}^m(q_{t_l-m^\kappa},y)} \mu(dy) \right],
\end{align}
and hence
\begin{align}\label{delta_F}
& |\mathcal{F}^m_{l,t}(k)-\mathcal{F}^m_{l-1,t}(k)|\\
\leq &o(1)+\left|E\bigg[\tilde X^m\left( \int  e^{ik_l\cdot  G_{t_l}^m(q_{t_l-m^\kappa},y)} \mu(dy) - e^{-k_l\cdot M(t_l,q_{t_l})k_l/2}\right)\bigg]\right|.\notag\\
\leq &o(1)+E\bigg[\left| \int  e^{ik_l\cdot  G^m_{t_l}(q_{t_l-m^\kappa},y)} \mu(dy) - e^{-k_l\cdot M(t_l,q_{t_l})k_l/2}\right|\bigg]\notag\\
\leq &o(1)+E\bigg[\int \left|  e^{ik_l\cdot  G^m_{t_l}(q_{t_l-m^\kappa},y)} - e^{ik_l\cdot  G^m_{t_l}(q_{t_l},y)} \right| \mu(dy)\bigg]\notag\\
&+E\bigg[\left| \int  e^{ik_l\cdot  G^m_{t_l}(q_{t_l},y)} \mu(dy) - e^{-k_l\cdot M(t_l,q_{t_l})k_l/2}\right|\bigg]\notag\\
\leq &o(1)+\|k_l\|E\bigg[\int \left\|  G^m_{t_l}(q_{t_l-m^\kappa},y) -G^m_{t_l}(q_{t_l},y)\right\| \mu(dy)\bigg]\notag\\
&+E\bigg[\left| \int  e^{ik_l\cdot  G^m_{t_l}(q_{t_l},y)} \mu(dy) - e^{-k_l\cdot M(t_l,q_{t_l})k_l/2}\right|\bigg]\notag.
\end{align}

Using \req{G_q_Lipschitz}, Lemma \ref{q_cont_lemma}, and the Burkholder-Davis-Gundy inequality  we have
\begin{align}\label{E_delta_G}
&E\bigg[\int \left\|  G^m_{t_l}(q_{t_l-m^\kappa},y) -G^m_{t_l}(q_{t_l},y)\right\| \mu(dy)\bigg]\\
\leq & \left(\tilde C\| \gamma\sigma\|_\infty  \int_0^\infty r e^{-\lambda r}dr +\frac{\tilde C}{\lambda }\right)m^{-1/2}E\left[ \|q_{t_l-m^\kappa}-q_{t_l}\|\right]\notag\\
&\times\int  \sup_{t_l-m^\kappa\leq s\leq t_l}\|y(m^\kappa)-y(s-(t_l-m^\kappa))\| \mu(dy)\notag\\
\leq &\tilde C m^{-1/2}(m^\kappa+m^{\kappa/2}) E\left[  \sup_{t_l-m^\kappa\leq s\leq t_l}\|W_{t_l}-W_s\| \right]\notag\\
\leq &\tilde C m^{-1/2}(m^\kappa+m^{\kappa/2}) m^{\kappa/2}.\label{E_delta_G2}
\end{align}
\begin{comment}
\begin{align}
&E\left[  \sup_{t_l-m^\kappa\leq s\leq t_l}\|W_{t_l}-W_s\| \right]\\
\leq &2 E\left[  \sup_{t_l-m^\kappa\leq s\leq t_l}\|W_{s}-W_{t_l-m^\kappa}\| \right]\\
\leq &2 E\left[  \sup_{t_l-m^\kappa\leq s\leq t_l}\|\int_{t_l-m^\kappa}^s dW_r\|^2 \right]^{1/2}\\
\leq  &C E\left[  \sup_{t_l-m^\kappa\leq s\leq t_l}\int_{t_l-m^\kappa}^s \|I\|dr \right]^{1/2}\\
=  &C  m^{\kappa/2}\\
\end{align}
\end{comment}

If we choose $\kappa\in(1/2,1)$ then \req{E_delta_G2} is $o(1)$ and hence
\begin{align}\label{delta_F2}
& |\mathcal{F}^m_{l,t}(k)-\mathcal{F}^m_{l-1,t}(k)|\\
\leq &o(1)+E\bigg[\left| \int  e^{ik_l\cdot  G^m_{t_l}(q_{t_l},y)} \mu(dy) - e^{-k_l\cdot M(t_l,q_{t_l})k_l/2}\right|\bigg]\notag.
\end{align}

The expression inside the expected value in \req{delta_F2} is bounded by $2$, and so if we can show that
\begin{align}\label{fix_q_result}
\lim_{m\rightarrow 0}  \int  e^{ik\cdot  G_t^m(q,y)} \mu(dy)=e^{-k\cdot M(t,q)k/2}
\end{align}  
for every $q,k\in\mathbb{R}^n$ then  the result will follow from the dominated convergence theorem.

Following the calculations of Section \ref{sec:simp_seq} again, now with time and state-independent $\sigma=\sigma(t,q)$ and $\tilde \gamma=\tilde\gamma(t,q)$ ($t,q$ fixed) we find that
\begin{align}
\lim_{m\rightarrow 0}  \bigg|&\int  e^{ik\cdot  G^m_t(q,y)} \mu(dy)\\
&- E\left[\exp\left( \frac{1}{\sqrt{m}}   ik\cdot e^{-\tilde \gamma(t,q) t/m}\int  e^{\tilde \gamma(t,q)s/m}\sigma(t,q) dW_s\right)\right]\bigg|=0.\notag
\end{align}
  The state-independent result, Lemma \ref{state_ind_lemma},  implies that
\begin{align}
&\lim_{m\rightarrow 0} E\left[\exp\left( \frac{1}{\sqrt{m}}  ik\cdot e^{-\tilde \gamma(t,q) t/m}\int  e^{\tilde \gamma(t,q)s/m}\sigma(t,q) dW_s\right)\right]\\
=& \frac{1}{(2\pi)^{n/2}|\det M(t,q)|^{1/2}}\int e^{ik\cdot z} e^{-z\cdot M(t,q)^{-1}z/2}dz\notag\\
=&e^{-k\cdot M(t,q)k/2}\notag
\end{align}  
for each $t$ and $q$.  This completes the proof.

\begin{comment}
The previous estimates imply that the difference between 
\begin{align}
 z^m_{1,t}=\frac{1}{\sqrt{m}} e^{-\tilde\gamma(t,q) t/m} \int_0^{t}e^{\tilde\gamma(t,q)s/m} \sigma(t,q) dW_s
\end{align}
and 
\begin{align}\label{simp7}
z^m_{7,t}\equiv& \frac{1}{{m}^{3/2}}\int_{t-m^\kappa}^te^{-\tilde\gamma(t,q)(t-s)/m} \tilde \gamma(t,q) \sigma(t,q) (W_t-W_s) ds\\
=&G_t^m(q,W^{m,t})
\end{align}
goes to zero in $L^1$.

So
\begin{align}
\lim_{m\to 0} |E[\exp(ik\cdot z^m_{1,t})]- E[\exp(ik\cdot z^m_{7,t})]|=0
\end{align}
Hence
\begin{align}
&\lim_{m\to 0} E[\exp(ik\cdot z^m_{1,t})]=\int\exp(ik\cdot G_t^m(q,y))d\mu(y)
\end{align}
if either exists.
\end{comment}

\end{proof}

\begin{corollary}\label{corr:fluc_dis}
When a fluctuation dissipation relation holds pointwise for a time and position dependent ``temperature" $T(t,q)$, i.e.
\begin{align}\label{fluc_dis}
\Sigma_{ij}(t,q)=2k_BT(t,q) \gamma_{ij}(t,q),
\end{align}
then
\begin{align}
M(t,q)=k_BT(t,q)\equiv 1/\beta(t,q),
\end{align}
and so the limiting distribution is
\begin{align}
d\nu = \left(\prod_{j=1}^N\left(\frac{\beta(t_j,q_j)}{2\pi}\right)^{n/2}\exp\left[- \beta(t_j,q_j) \|z_j\|^2/2\right]dz_j\right) \mu_{J,t}(dJ,dq_1,...,dq_N).
\end{align}
\end{corollary}
\noindent As stated earlier in Section \ref{sec:summary}, here we recognize the Gibbs measure for the $z$-variables, and hence can interpret this result as expressing an instantaneous equilibration of the scaled momentum variables (in particular, of the kinetic energy)  in the limit $m\to 0$.

\section{A Stronger Convergence Result when $N=1$}\label{sec:conv_rate}
When $N=1$, and under stronger assumptions on $J^m$ and $J$, the estimates from Section \ref{sec:limit_dist} will allow us to prove convergence of $E[h(Y^m)]$ as $m\rightarrow 0$ for a wider class of functions than just bounded continuous ones (in which case convergence is guaranteed by weak convergence of the distributions of the $Y^m$). This will also provide a bound on the convergence rate.  

Extending the class of functions in this way is significant as there are important physical quantities, such as the kinetic energy, that are not bounded functions of  $z_t^m$. This is relevant for the study of entropy production \cite{Chetrite2008,gawedzki2013fluctuation,Birrell_entropy}.

\subsection{Strategy for Bounding the Convergence Rate}
 Our goal is now to consider the $m\to 0$ limit of the processes
\begin{align}
Y^m_t=(J^m_t,q^m_t,z^m_t)
\end{align}
for $t>0$,  where $J_t^m$ are continuous, $\mathbb{R}^d$-valued, $\mathcal{F}_t$-adapted processes that share several important properties with $q_t^m$:\\
\begin{assumption}\label{assump:J_t}
We assume that for any $T>0$, $p>0$ we have
\begin{align}\label{J_t_limit}
 \sup_{t\in[0,T]}E\left[\|J_t^m-J_t\|^p\right]^{1/p}=O(m^{1/2}) 
\end{align}
as $m\rightarrow 0$, where $J_t$ is also a continuous,  $\mathcal{F}_t$-adapted process. We also assume that $J_t$ has the same boundedness property as $q_t$:
\begin{align}\label{J_t_Lp_bound}
E\left[\sup_{t\in[0,T]}\|J_t\|^p\right]^{1/p}<\infty
\end{align}
for all $T>0$, $p>0$, as well as the same $L^p$-continuity property (see Lemma \ref{q_cont_lemma}):

For any $T>0$, $p>0$ there exists $\tilde C>0$ such that for any $0\leq s\leq t\leq T$  we have
\begin{align}\label{J_t_cont_assump}
E\left[  \|J_t -J_s\|^p\right]^{1/p}\leq\tilde C\left((t-s)+\left(t-s \right)^{1/2}\right).
\end{align}
\end{assumption}
As discussed in Assumption \ref{assump:J}, we still have in mind processes such as $J_t^m=\int_0^t g(r,q_r^m)dr$, but are purposely general about the nature of the processes $J_t^m$ and $J_t$ here.

Fix $q,K>0$ and let $h:\mathbb{R}^{d+2n}\rightarrow\mathbb{C}$ be any  $C^1$ function that satisfies
\begin{align}\label{nabla_h_bound}
\|\nabla h(y)\|\leq K(1+\|y\|^q) 
\end{align}
and consider
\begin{align}
H_t^m\equiv E[h( Y_t^m)].
\end{align}

Similarly to Section \ref{sec:simp_seq}, we will compute $\lim_{m\rightarrow 0} H_t^m$ by showing that, if it exists, it is equal to the limits of a sequence of related quantities of the form $H^m_{l,t}\equiv E[h( Y^m_{l,t})]$.  Eventually we will arrive at a reduced form for which we can compute the limit explicitly.  The following lemma will be key to all of these reduction steps.  The intuition behind what we do here is the same as in Section \ref{sec:simp_seq}, but now we need to be more careful about the dependence on constants, hence the reason for our extra precision.

\begin{lemma}\label{reduction_lemma}
Let  $h:\mathbb{R}^{\tilde k}\rightarrow\mathbb{C}$ be any  $C^1$ function whose first derivative is polynomially bounded (\req{nabla_h_bound}). Suppose we have families of random variables $\tilde  Y^m$ and  $\hat  Y^m$ and some $\delta>0$ such that for every $p>0$
\begin{align}
E[\|\tilde Y^m\|^p]^{1/p}=O(1),  \hspace{2mm} E[ \|\tilde Y^m-\hat {Y}^m\|^p]^{1/p}=O(m^\delta),
\end{align}
as $m\rightarrow 0$.  Then $E[h(\tilde Y^m)]$ and $E[h(\hat Y^m)]$ exist for all sufficiently small $m$,
\begin{align}\label{dh}
E[\|\hat Y^m\|^p]^{1/p}=O(1), \text{ and  }\, |E[h(\tilde Y^m)]- E[h(\hat Y^m)]|=KO(m^\delta)
\end{align}
as $m\rightarrow 0$, where the implied constant in \req{dh} is independent of $K$ and of the choice of $h$ satisfying \req{nabla_h_bound}.
\end{lemma}
\begin{proof}
H\"older's inequality yields
\begin{align}
E[\|\hat{Y}^m\|^p]^{1/p}\leq E[\|\tilde{Y}^m\|^p]^{1/p}+E[\|\hat{Y}^m-\tilde Y^m\|^p]^{1/p}=O(1).
\end{align}
For any $y_1,y_2$ we have
\begin{align}\label{FTCalc}
h(y_2)-h(y_1)=\int_0^1 \nabla h(sy_2+(1-s)y_1)ds \cdot (y_2-y_1),
\end{align}
hence
\begin{align}
E\left[|h(\tilde Y^m)|\right]\leq |h(0)|+KE\left[ (1+\|\tilde Y^m\|^q)\|\tilde Y^m\|\right]=O(1).
\end{align}
Therefore  $E[h(\tilde Y^m)]$ exists for all sufficiently small $m$, and similarly for $E[h(\hat Y^m)]$.

Let $p,\tilde p>1$ be conjugate exponents.  Using \req{FTCalc} and  H\"older's inequality
\begin{align}
&|E[h(\tilde Y^m)]- E[h(\hat Y^m)]|\\
\leq&E\left[ \int_0^1 \|\nabla h(s \tilde Y^m+(1-s)\hat {Y}^m)\| \|\tilde Y^m-\hat {Y}^m\|ds\right]\notag\\
\leq& \int_0^1 E[\|\nabla h(s \tilde Y^m+(1-s)\hat {Y}^m)\|^{\tilde p}]^{1/\tilde p}dsE[ \|\tilde Y^m-\hat {Y}^m\|^p]^{1/p}.\notag
\end{align}
We have
\begin{align}
&\int_0^1 E[\|\nabla h(s \tilde Y^m+(1-s)\hat {Y}^m)\|^{\tilde p}]^{1/\tilde p}ds\\
\leq&K\int_0^1 E[(1+ \|s \tilde Y^m+(1-s)\hat {Y}^m\|^q)^{\tilde p}]^{1/\tilde p}ds\notag\\
\leq& K\left(1+\int_0^1E[\|\tilde Y^m+(1-s)(\hat {Y}^m-
\tilde Y^m)\|^{\tilde pq}]^{1/\tilde p}ds\right)\notag\\
\leq& K\left(1+\tilde C\int_0^1E[\|\tilde Y^m\|^{\tilde pq}]^{1/\tilde p}+(1-s)^qE[\|\hat {Y}^m-
\tilde Y^m\|^{\tilde pq}]^{1/\tilde p}ds\right)\notag\\
\leq& K\left(1+\tilde C\left(E[\|\tilde Y^m\|^{\tilde pq}]^{1/\tilde p}+E[\|\hat {Y}^m-
\tilde Y^m\|^{\tilde pq}]^{1/\tilde p}\right)\right)\notag\\
=&KO(1).\notag
\end{align}
Hence
\begin{align}
|E[h(\tilde Y^m)]- E[h(\hat Y^m)]|\leq KO(1)O(m^\delta).
\end{align}
\end{proof}

For $Y_t^m=(J_{t}^m,q^m_t,z_{t}^m)$, note that Assumptions \ref{assump:conv} and \ref{assump:J_t} imply that  $E[\|\tilde Y_t^m\|^p]^{1/p}=O(1)$  for any $p>0$. Therefore, if we can find $\delta> 0$ and a finite sequence of (``simplified") random variables  $Y^m_{l,t}$, $l=1,...,k$ (and letting $Y^m_{0,t}\equiv Y^m_t$) such that  
\begin{align}\label{q_J_Lp_conv}
 E[ \|Y^m_{l-1,t}-Y^m_{l,t}\|^p]^{1/p}=O(m^\delta)
\end{align}
for all $p>0$, $l=1,...,k$, then  iterating Lemma \ref{reduction_lemma} $k$ times will allow us to conclude that $E[h( Y^m_{t})]$ and $E[h( Y^m_{k,t})]$ exist for all sufficiently small $m$ and
\begin{align}
|E[h( Y_t^m)]-E[h( Y^m_{k,t})]|=KO(m^\delta)
\end{align}
as $m\rightarrow 0$, where the implied constant is independent of $K$ and of $h$. Note that by H\"older's inequality, it suffices to show \req{q_J_Lp_conv} for all $p$ larger than some $p_0$.

If in addition, we can show that $\lim_{m\rightarrow 0} E[h( Y^m_{k,t})]\equiv H_t$ exists
and
\begin{align}
E[h(Y^m_{k,t})]=H_t+KO(m^\delta)  
\end{align}
 then we can further conclude that $\lim_{m\rightarrow 0} E[h( Y_t^m)]=H_t$  and
\begin{align}
 E[h( Y_t^m)]=H_t+KO(m^\delta),
\end{align}
 thereby accomplishing our goal. Note that in the above argument, the implied constant in the big-O notation can be chosen independent of $K$ and $h$. The reason for the careful attention we have paid to the dependence on $K$ and $h$, both here and in the sequel, will become clear as we proceed.

\subsection{The Convergence Rate Bound}

\begin{theorem}\label{thm:conv_dis2}
 Let   $K,q>0$, $0<\delta<1/2$, and $h:\mathbb{R}^{d+2n}\rightarrow\mathbb{C}$ be a  $C^1$ function that satisfies
\begin{align}
\|\nabla h(y)\|\leq K(1+\|y\|^q).
\end{align}
Define
\begin{align}
M(t,q)=\int_0^\infty  e^{-\tilde\gamma(t,q) \zeta}\Sigma(t,q) e^{-\tilde\gamma^T(t,q) \zeta}d\zeta
\end{align}
and
\begin{align}
H_t=E\left[ \frac{1}{(2\pi)^{n/2} |\det M(t,q_t)|^{1/2}} \int  h(J_t,q_t,z) e^{-z\cdot M^{-1}(t,q_t)z/2}dz\right].
\end{align}
\begin{comment}
Its not hard to show this exists by using the bound on $\nabla h$
\end{comment}
 Then, for any $\tilde q>0$,
\begin{align}
&E\left[h(J_t^m,q_t^m,z_t^m)\right]=H_t+KO(m^{\delta})+h(0)O(m^{\tilde q})
\end{align}
as $m\rightarrow 0$, where the implied constants are independent of $K$ and $h$.
\end{theorem}
\begin{proof}

Fix $\delta\in (0,1/2)$ and let $\kappa\in(2\delta,1)$. Lemma \ref{reduction_lemma} together with the estimates of Section \ref{sec:simp_seq} and Assumptions \ref{assump:conv} and \ref{assump:J_t} reduce the problem to computing $\lim_{m\rightarrow 0}H^m_{7,t}\equiv H_t$ where
\begin{align}\label{7H_def}
H^m_{7,t}\equiv E\left[h \left(J_{t-m^\kappa},q_{t-m^\kappa}, G_t^m(q_{t-m^\kappa},W^{m,t}\right)\right],
\end{align}
where $W^{m,t}$  and $G_t^m$ have been defined in  \req{W_mt_def} and \req{G_mt_def} respectively.

If we can show that $\kappa\in (2\delta,1)$ can be chosen so that $|H^m_{7,t}- H_t|=KO(m^\delta)$ (with the implied constant independent of $K$ and $h$) then, as discussed above, we will be able to conclude that $\lim_{m\rightarrow 0} E[h( Y_t^m)]=H_t$  and
\begin{align}
 E[h( Y_t^m)]=H_t+KO(m^\delta),
\end{align}
 again with the implied constant independent of $K$ and $h$.

As in the proof of Theorem \ref{thm:conv_dis}, see \req{independence_calc}, we can use independence to write
\begin{align}
H^m_{7,t}=E\left[  \int h(J_{t-m^\kappa},q_{t-m^\kappa},G^m_t(q_{t-m^\kappa},y)) \mu(dy)\right]
\end{align}
where $\mu$ is the Wiener measure on $C([0,1],\mathbb{R}^n)$.

We now begin computing the desired bound.
\begin{align}\label{delta_H}
&|H^m_{7,t}-H_t|\\
\leq& E\left[ \int\left| h(J_{t-m^\kappa},q_{t-m^\kappa},G^m_t(q_{t-m^\kappa},y)) - h(J_{t},q_t,G^m_t(q_{t},y))\right| d\mu(y)  \right]\notag\\
&+\left|E\left[ \int h(J_{t},q_t,G^m_t(q_{t},y)) d\mu(y) \right.\right.\notag\\
&\left.\left.-\frac{1}{(2\pi)^{n/2} |\det M(t,q_t)|^{1/2}} \int  h(J_t,q_t,z) e^{-z\cdot M^{-1}(t,q_t)z/2}dz\right]\right|  .\notag
\end{align}

To bound the first term, given conjugate exponents $p,\tilde p>1$, we employ a similar calculation to Lemma \ref{reduction_lemma} along with the estimates \req{G_bound} and \req{G_q_Lipschitz}. Using the notation $Q_{J,t}=(J_t,q_t)$ and letting $\tilde C$ denote a constant that potentially varies line to line (but does not depend on $m$, $K$, or $h$) we have
\begin{align}
&  E\left[ \int\left| h(Q_{J,t-m^\kappa},G^m_t(q_{t-m^\kappa},y)) - h(Q_{J,t},G^m_t(q_{t},y))\right| d\mu(y)  \right]\\
\leq& \tilde C K\int \left( 1+E[\|Q_{J,t}\|^{\tilde pq}]^{1/\tilde p}+E[\|Q_{J,t-m^\kappa}-
Q_{J,t}\|^{\tilde pq}]^{1/\tilde p}+E[\|G^m_t(q_{t},y)\|^{\tilde pq}]^{1/\tilde p}\right.\notag\\
&\left.+E[\|G^m_t(q_{t-m^\kappa},y)-G^m_t(q_{t},y)\|^{\tilde pq}]^{1/\tilde p}\right)\notag\\
&\times \left(E[ \|Q_{J,t}-Q_{J,t-m^\kappa}\|^p]^{1/p}+E[\|G^m_t(q_{t},y)- G^m_t(q_{t-m^\kappa},y)\|^p]^{1/p}\right) d\mu(y)\notag\\
\leq& \tilde CK\left( \int \left( 1+E[\|Q_{J,t}\|^{\tilde pq}]^{1/\tilde p}+E[\|Q_{J,t-m^\kappa}-
Q_{J,t}\|^{\tilde pq}]^{1/\tilde p}+E[\|G^m_t(q_{t},y)\|^{\tilde pq}]^{1/\tilde p}\right.\right.\notag\\
&\left.\left.+E[\|G^m_t(q_{t-m^\kappa},y)-G^m_t(q_{t},y)\|^{\tilde pq}]^{1/\tilde p}\right)^2d\mu(y)\right)^{1/2}\notag\\
&\times \left(\int\left(E[ \|Q_{J,t}-Q_{J,t-m^\kappa}\|^p]^{1/p}+E[\|G^m_t(q_{t},y)- G^m_t(q_{t-m^\kappa},y)\|^p]^{1/p}\right)^2 d\mu(y)\right)^{1/2}\notag
\end{align}
\begin{align}
\leq& \tilde CK\bigg(   1+E[\|Q_{J,t}\|^{\tilde pq}]^{1/\tilde p}+E[\|Q_{J,t-m^\kappa}-
Q_{J,t}\|^{\tilde pq}]^{1/\tilde p}+\left(\int E[\|G^m_t(q_{t},y)\|^{\tilde pq}]^{2/\tilde p}d\mu(y)\right)^{1/2}\notag\\
&+\left(\int E[\|G^m_t(q_{t-m^\kappa},y)-G^m_t(q_{t},y)\|^{\tilde pq}]^{2/\tilde p}d\mu(y)\right)^{1/2}\bigg)\notag\\
&\times \left(E[ \|Q_{J,t}-Q_{J,t-m^\kappa}\|^p]^{1/p}+\left(\int E[\|G^m_t(q_{t},y)- G^m_t(q_{t-m^\kappa},y)\|^p]^{2/p}d\mu(y)\right)^{1/2}\right)\notag\\
\leq& \tilde CK\bigg(   1+(m^\kappa+m^{\kappa/2})^q+m^{-q/2}\left(\int \sup_{t-m^\kappa\leq s\leq t}\|y(m^\kappa)-y(s-(t-m^\kappa))\|^{2q}d\mu(y)\right)^{1/2}\\
&+m^{-q/2}E[\|q_t-q_{t-m^\kappa}\|^{\tilde pq}]^{1/\tilde p}\left(\int  \sup_{t-m^\kappa\leq s\leq t}\|y(m^\kappa)-y(s-(t-m^\kappa))\|^{2q}d\mu(y)\right)^{1/2}\bigg)\notag\\
&\times \bigg( m^{\kappa/2}+m^{\kappa}+m^{-1/2}E[\|q_t-q_{t-m^\kappa}\|^p]^{1/p}\notag\\
&\hspace{4cm}\times\left(\int \sup_{t-m^\kappa\leq s\leq t}\|y(m^\kappa)-y(s-(t-m^\kappa))\|^2d\mu(y)\right)^{1/2}\bigg).\notag
\end{align}

For any $\tilde q>1$ we have
\begin{align}
&\int  \sup_{t-m^\kappa\leq s\leq t}\|y(m^\kappa)-y(s-(t-m^\kappa))\|^{ \tilde q}d\mu(y)\\
=&E\left[\sup_{t-m^\kappa\leq s\leq t}\|W_t -W_{s}\|^{ \tilde q}\right]\notag\\
\leq &2^{q}E\left[\sup_{t-m^\kappa\leq s\leq t}\|W_s-W_{t-m^\kappa}\|^{\tilde q}\right]\notag\\
\leq &\tilde C m^{\tilde q\kappa/2}.\notag
\end{align}

Without loss of generality we can assume $2q>1$ and $m\leq 1$. Therefore
\begin{align}
&  E\left[ \int\left| h(Q_{J,t-m^\kappa},G^m_t(q_{t-m^\kappa},y)) - h(Q_{J,t},G^m_t(q_{t},y))\right| d\mu(y)  \right]\\
\leq & \tilde CK\left(   1+(m^\kappa+m^{\kappa/2})^q+m^{-(1-\kappa) q/2}(1+E[\|q_t-q_{t-m^\kappa}\|^{\tilde pq}]^{1/\tilde p})\right)\notag\\
&\times \left( m^{\kappa/2}+m^{\kappa}+m^{-(1-\kappa)/2}E[\|q_t-q_{t-m^\kappa}\|^p]^{1/p}\right)\notag\\
\leq & \tilde CK\left(   1+(m^\kappa+m^{\kappa/2})^q+m^{-(1-\kappa) q/2}(1+(m^{\kappa/2}+m^\kappa)^q)\right)\notag\\
&\times \left( m^{\kappa/2}+m^{\kappa}+m^{-(1-\kappa)/2}(m^{\kappa/2}+m^\kappa)\right)\notag\\
\leq & \tilde CKm^{-(1-\kappa) q/2} \left( m^{\kappa/2}+m^{\kappa-1/2}\right).\notag
\end{align}
Since $\kappa\in(2\delta,1)$ is arbitrary we can choose $\kappa>1/2$ sufficiently close to $1$ and obtain
\begin{align}
&  E\left[ \int\left| h( J_{t-m^\kappa}, q_{t-m^\kappa},G^m_t(q_{t-m^\kappa},y)) - h(J_{t},q_t,G^m_t(q_{t},y))\right| d\mu(y)  \right]=KO(m^\delta)
\end{align}
as $m\rightarrow 0$, where the implied constant is independent of $K$ and $h$.

Now focus on the second term in \req{delta_H},
\begin{align}
\Delta H_2\equiv |H^m_{8,t}-H_t|,
\end{align}
\begin{align}\label{8H_def}
 H^m_{8,t}\equiv E\left[ \int h(J_{t},q_t,G^m_t(q_{t},y)) d\mu(y)  \right].
\end{align}

Fixing $\omega\in\Omega$   and considering  $\int h(J_{t}(\omega),q_t(\omega),G^m_t(q_{t}(\omega),y)) d\mu(y)$, we see that  this is the same expression that one would obtain for $H^m_{8,t}$ had one been, from the beginning, working with fixed (i.e. time and state-independent) drag and diffusion $\tilde \gamma(t,q_t(\omega))$ and $\sigma(t,q_t(\omega))$, and the  function $\tilde h:\mathbb{R}^n\rightarrow\mathbb{C}$, $\tilde h(z)=h(Q_{J,t}(\omega),z)$, {\em and} on a different probability space with a Wiener process $\tilde W$, distinct from the $W$ used up to this point.  We denote the expected value with respect to this new probability measure by $\tilde E$.

  $\tilde h$ is $C^1$ and we have the bounds
\begin{align}
&\|\nabla \tilde h(z)\|\leq 2^qK(1+\|Q_{J,t}(\omega)\|^q +\|z\|^q)\leq 2^q K(1+\|Q_{J,t}(\omega)\|^q)(1+\|z\|^q)\\
&\equiv K^\prime (1+\|z\|^q),\notag\\
& |\tilde h(z)|\leq |\tilde h(0)|+ K^\prime(1+\|z\|^q)\|z\|\leq2\max\{ |\tilde h(0)|,K^\prime\}(1+\|z\|^{1+q})\\
&\equiv  \tilde K(1+\|z\|^{1+q}).\notag
\end{align}

Therefore applying our  arguments from Section \ref{sec:limit_dist}  to this system shows that   there exists $m_0$, $\tilde C$ independent of $m$, $\omega$, $h$, and $K$ such that for all $0<m\leq m_0$ we have
\begin{align}
&\bigg|\tilde E\left[\tilde h\left(\frac{1}{\sqrt{m}}  e^{-\tilde\gamma(t,q_t(\omega)) t/m} \int_0^t e^{\tilde\gamma(t,q_t(\omega)) s/m} \sigma(t,q_t(\omega))  d\tilde W_s\right)\right]\\
&-\int h(Q_{J,t}(\omega),G^m_t(q_{t}(\omega),y)) d\mu(y)\bigg|\notag\\
\leq&     K^\prime(\omega)\tilde C  m^{\kappa/2}\notag
\end{align}
and
\begin{align}
&\left|\tilde E\left[\tilde h\left(\frac{1}{\sqrt{m}}  e^{-\tilde\gamma(t,q_t(\omega)) t/m} \int_0^t e^{\tilde\gamma(t,q_t(\omega)) s/m} \sigma(t,q_t(\omega))  d\tilde W_s\right)\right]\right.\\
&\left.-\frac{1}{(2\pi)^{n/2} |\det M(t,q_t(\omega))|^{1/2}} \int  h(Q_{J,t}(\omega),z) e^{-z\cdot M^{-1}(t,q_t(\omega))z/2}dz \right|\notag\\
\leq&    \tilde K(\omega)\tilde C  m^{\tilde q},\notag
\end{align}
the latter by Lemma \ref{state_ind_lemma}. Note that the randomness in the above two expectations comes from $\tilde W_s$;  $\omega$ here is fixed and {\em not} integrated over in these expressions.

Without loss of generality, we can assume $\tilde q\geq \delta$. Therefore, for $0<m\leq m_0$ we have the bound
\begin{align}
\Delta H_2\leq&\tilde C  m^{\kappa/2} E[K^\prime]+\tilde C  m^{\tilde q}E[\tilde K]\\
\leq& E[K^\prime] O(m^\delta)+E[|\tilde h(0)|]O(m^{\tilde q})\notag\\
\leq& K(1+E[\|Q_{J,t}\|^q]) O(m^\delta)+E[| h(0)|+K(1+\|Q_{J,t}\|^q)\|Q_{J,t}\|]O(m^{\tilde q})\notag\\
=&KO(m^{\delta})+|h(0)|O(m^{\tilde q}).\notag
\end{align}
The implied constants are independent of $K$ and $h$ and are finite by \req{q_Lp_bound} and \req{J_t_Lp_bound}, so this completes the proof.

\end{proof}

As remarked in the introduction, an important step in \cite{Hottovy2014} consists of showing that the kinetic energy $m\|v^m_t\|^2$ is of order one, suggesting that $v^m_t$ diverges as ${1 \over \sqrt{m}}$. We end by using the above theorem to prove a result that further supports this intuition.  We show that, in probability and for any $q\in(0,1/2)$, $\|u_t^m\|$ grows faster than $1/m^q$ as $m\to 0$.
\begin{corollary}
Let $t>0$, $R> 0$, and $q\in (0,1/2)$.  Then
\begin{align}
\lim_{m\to 0}P(\|u_t^m\|\leq R/m^q)=0.
\end{align}
\end{corollary}
\begin{proof}
Let $\epsilon>0$ and $\phi$ be a smooth bump function for the ball $B_1(0)\subset \mathbb{R}^n$ with support in $B_2(0)$ and values in $[0,1]$. Given $\delta>0$ let $\phi_\delta(z)=\phi(z/\delta)$.  Then for any $0<m\leq\epsilon$,
\begin{align}
&P(\|u_t^m\|\leq R/m^q)=P(\|z_t^m\|\leq m^{1/2-q}R)\\
\leq& P(\|z_t^m\|\leq \epsilon ^{1/2-q}R)\leq E\left[\phi_{\epsilon^{1/2-q}R}(z_t^m)\right].\notag
\end{align}
$\phi_{\epsilon^{1/2-q}R}$ is smooth with bounded derivative, therefore Theorem \ref{thm:conv_dis2} implies
\begin{align}
&\lim_{m\rightarrow 0}E\left[\phi_{\epsilon^{1/2-q}R}(z_t^m)\right]\\
=&E\left[ \frac{1}{(2\pi)^{n/2} |\det M(t,q_t)|^{1/2}} \int  \phi_{\epsilon^{1/2-q}R}(z) e^{-z\cdot M^{-1}(t,q_t)z/2}dz\right].\notag
\end{align}
Hence
\begin{align}
&\limsup_{m\to 0}P(\|u_t^m\|\leq R/m^q)\\
\leq& E\left[ \frac{1}{(2\pi)^{n/2} |\det M(t,q_t)|^{1/2}} \int  \phi_{\epsilon^{1/2-q}R}(z) e^{-z\cdot M^{-1}(t,q_t)z/2}dz\right]\notag
\end{align}
for all $\epsilon>0$.  The right hand side converges to $0$ as $\epsilon\to0$ by the dominated convergence theorem, thus proving the claim.

\end{proof}

\appendix

\numberwithin{assumption}{section}
\section{Material from \cite{BirrellHomogenization}}\label{app:assump}
\setcounter{assumption}{0}
    \renewcommand{\theassumption}{\Alph{section}\arabic{assumption}}

    \renewcommand{\thelemma}{\Alph{section}\arabic{lemma}}

In this appendix, we give a list of properties  that, as shown in \cite{BirrellHomogenization}, are sufficient to  guarantee that the solutions to the SDE \req{q_eq}-\req{u_eq} satisfy the conditions from  Assumption \ref{assump:conv}. The assumptions listed here that are explicitly used in the current paper are referenced when they are needed. We collect the full set of assumptions here only for reference.

We assume that
\begin{enumerate}
\item There exists $a,b\geq 0$ s.t. $\tilde V(t,q)\equiv a+b\|q\|^2+V(t,q)$ is non-negative for all $t,q$.
\item There exists $M,C> 0$ s.t. $|\partial_tV(t,q)|\leq M+C(\|q\|^2+\tilde V(t,q))$ for all $t,q$.
\item There exists $C>0$ such that the (random) initial conditions satisfy $ \|u^m_0\|^2 \leq Cm$ for all $m>0$ and all $\omega\in\Omega$ and $E[\|q^m_0\|^p]<\infty$, $E[\|q_0\|^p]<\infty$, and $E[\|q_0^m-q_0\|^p]^{1/p}=O(m^{1/2})$ for all $p>0$.
\item $\partial_t\psi$, $\tilde F$, and $\sigma$ are continuous and bounded.
\item  $\gamma$ is symmetric with eigenvalues  bounded below by some $\lambda>0$.
\end{enumerate}
We also assume that, for every $T>0$, the following hold uniformly for $(t,x)\in [0,T]\times\mathbb{R}^{2n}$:
\begin{enumerate}
\item  $\gamma$ and $\tilde F$ are continuous and bounded.
\item $\nabla_q V$, $\tilde F$, and $\sigma$ are Lipschitz in $x$ uniformly in $t$.
\item $V$ is $C^2$, $\gamma$ is $C^2$, $\psi$ is $C^3$, and the following are bounded:\\
$\nabla_qV$, $\partial_t\psi$, $\partial_{q^i}\psi$, $\partial_{q^i}\partial_{q^j}\psi$, $\partial_t\partial_{q^i}\psi$, $\partial_t\partial_{q^j}\partial_{q^i}\psi$, $\partial_{q^l}\partial_{q^j}\partial_{q^i}\psi$, $\partial_t\gamma$, $\partial_{q^i} \gamma$, $\partial_t\partial_{q^j}\gamma$,  $\partial_{q^i}\partial_{q^j}\gamma$.
\end{enumerate}
Note that these properties imply that $\tilde\gamma$ (see \req{tilde_gamma_def}), $\tilde\gamma^{-1}$, $\partial_t\tilde\gamma^{-1}$, $\partial_{q^i}\tilde\gamma^{-1}$, $\partial_t\partial_{q^j}\tilde\gamma^{-1}$, and  $\partial_{q^i}\partial_{q^j}\tilde\gamma^{-1}$ are bounded on $ [0,T]\times\mathbb{R}^{n}$ as well.

\section{Required Lemmas}

\setcounter{lemma}{0}
    \renewcommand{\thelemma}{\Alph{section}\arabic{lemma}}

We need the following lemma bounding the spectrum of a matrix.  See, for example, Appendix A in \cite{BirrellHomogenization} for a proof.
\begin{lemma}\label{eig_bound_lemma1}
Let $A$ be an $n\times n$ real or complex matrix with symmetric part $A^s=\frac{1}{2}(A+A^*)$. If the eigenvalues of $A^s$ are bounded above (resp. below) by $\alpha$ then the real parts of the eigenvalues of $A$ are bounded above (resp. below) by $\alpha$.
\end{lemma}

The following lemma shows, in particular, that on compact time intervals the limit process, $q_t$, is $1/2$-H\"older continuous in $t$ in the $L^p$-norm.
\begin{lemma}\label{q_cont_lemma}
Let $q_t$ be the solution to \req{q_SDE} and $T>0$. Then for any $p>0$ there exists $\tilde C>0$ such that for any $0\leq s\leq t\leq T$  we have
\begin{align}
E\left[  \|q_t -q_s\|^p\right]^{1/p}\leq\tilde C\left((t-s)+\left(t-s \right)^{1/2}\right)
\end{align}
where $\tilde C$ depends only on $T$, $p$, $n$, and the drift vector field and diffusion matrix of the SDE for $q_t$.
\end{lemma}
\begin{proof}
By \req{q_SDE},
\begin{align}
q_t-q_s=& \int_s^t \tilde \gamma^{-1}(r,q_r)F(r,q_r,\psi(r,q_r))+S(r,q_r)dr+\int_s^t\tilde \gamma^{-1}(r,q_r)\sigma(r,q_r) dW_r.
\end{align}
Therefore, using boundedness of $\tilde\gamma$, $F$, $S$, $\sigma$ on $[0,T]\times\mathbb{R}^{2n}$ along with the  Burkholder-Davis-Gundy inequality, for $p>1$ we have
\begin{align}
E[\|q_t-q_s\|^p]^{1/p}\leq& \tilde C(t-s)+E\left[\left\|\int_s^t\tilde \gamma^{-1}(r,q_r)\sigma(r,q_r) dW_r\right\|^p\right]^{1/p}\\
\leq& \tilde C\left((t-s)+E\left[\left(\int_s^t\|\tilde \gamma^{-1}(r,q_r)\sigma(r,q_r))\|^2dr \right)^{p/2}\right]^{1/p}\right)\notag\\
\leq& \tilde C\left((t-s)+\left(t-s \right)^{1/2}\right).\notag
\end{align}
The result for all $p>0$ then follows from an application of H\"older's inequality.
\end{proof}

We will need the following bound on the difference between the fundamental solutions corresponding to two  linear ODEs.
\begin{lemma}\label{fund_matrix_diverg}
Let $B_i:[0,T]\rightarrow\mathbb{R}^{n\times n}$ be continuous and suppose their symmetric parts have eigenvalues bounded above by $\lambda$, uniformly in $t$.  Consider the fundamental solutions $\Phi_i^\prime(t)=B_i(t)\Phi_i(t)$, $\Phi_i(0)=I$.  Then for any $0\leq t\leq T$ we have the bound
\begin{align}
\|\Phi_1(t)-\Phi_2(t)\|\leq \int_0^t \|B_1(s)-B_2(s)\| ds e^{\lambda t} .
\end{align}
\end{lemma}
\begin{proof}
Define $Y=\Phi_1-\Phi_2$.
\begin{align}
Y^\prime=B_1\Phi_1-B_2\Phi_2=B_1 Y+(B_1-B_2)\Phi_2,\,\, \, Y(0)=0,
\end{align}
hence 
\begin{align}
Y(t)=\Phi_1(t)\int_0^t\Phi_1^{-1}(s)(B_1(s)-B_2(s))\Phi_2(s)ds.
\end{align}
Therefore
\begin{align}
\|Y(t)\|\leq \int_0^t e^{\lambda(t-s)} \|B_1(s)-B_2(s)\| e^{\lambda s}ds= \int_0^t \|B_1(s)-B_2(s)\| ds e^{\lambda t} .
\end{align}
\end{proof}

The following result concerns the distribution of certain integrals with respect to a Wiener process.
\begin{lemma}\label{lemma:exp_mart}
Let $A:[0,\infty)\rightarrow \mathbb{R}^{m\times n}$ be in $L^2_{loc}$. Then $X_t=\int_{0}^t A(s)dW_s$  is a continuous martingale and each $X_t$ is normally distributed with mean zero and covariance matrix $C_t^{ij}=\int_{0}^t A^i_k(s)\delta^{kl} A_l^j(s)ds$.
\end{lemma}
\begin{proof}
All the  integrals exist so $X_t$ is a well defined continuous local martingale.  The quadratic covariation is 
\begin{align}
[X]^{ij}_t=\int_{0}^t A^i_k(s)  \delta^{kl} A^j_l(s) ds.
\end{align}
$E[[X]^{ii}_t]<\infty$ so $X_t$ (see problem 5.24 on p.38 of \cite{karatzas2014brownian}) is a martingale and we can construct the complex exponential martingale
\begin{align}
Z_t=&\exp(ik\cdot  X_t+k\cdot [X]_t\cdot k/2)\\
=&\exp\left(ik \cdot X_t+\frac{1}{2} \int_{0}^t k_i  A^i_k(s) \delta^{kl} A_l^j(s) k_jds\right).\notag
\end{align}
We note that this is a local martingale because
\begin{align}
Z_t=1+i\int_0^t Z_s d(k\cdot X_s)
\end{align}
and is in fact a martingale since
\begin{align}
E\left[\int_0^t |Z_s|^2 d[k\cdot X]_s\right]=\int_0^t \exp(k\cdot[X]_s\cdot k/2)   k_iA^i_k(s)  \delta^{kl} A^j_l(s) k_j ds<\infty.
\end{align}

The Fourier transform of the distribution of $X_t$ is
\begin{align}
&E[e^{i k\cdot X_t} ]=E[Z_t]  \exp\left(-\frac{1}{2} \int_{0}^t k_i  A^i_k(s) \delta^{kl} A_l^j(s) k_jds\right)\\
=&\exp(-\frac{1}{2}k\cdot C_t\cdot k),\label{FT_X}
\end{align}
where we used the fact that $Z_t$ is a martingale and $Z_0=1$, hence $E[Z_t]=1$.  \req{FT_X} equals the Fourier transform of the normal distribution with mean zero and covariance $C_t$, thereby proving the claim.
\end{proof}

Finally, we need a  martingale representation result for initially enlarged filtrations.  Its proof will rely on the following density lemma.

\begin{lemma}\label{density_lemma}
Let $(\Omega,\mathcal{F},P)$ be a  probability space and $X_1,X_2$ be  random quantities on $\Omega$, valued in measurable spaces $(\mathcal{X}_i,\mathcal{M}_i)$, $i=1,2$ resp.  Define
\begin{align}
D=\{1_{C_1}1_{C_2}:C_i\in\sigma(X_i)\}.
\end{align}
The span of $D$ is dense in $L^2(\sigma(X_1,X_2),P)$ (and hence in $L^2(\overline{\sigma(X_1,X_2)},P)$ as well).
\end{lemma}
\begin{proof}

We need to show that if $g\in L^2(\sigma(X_1,X_2),P)$ is orthogonal to every element of $D$ then $g=0$.\\

  For $g\in L^2(\sigma(X_1,X_2),P)$ there exists a $\mathcal{M}_1\bigotimes\mathcal{M}_2$-measurable $\tilde g$ s.t. $g=\tilde g\circ (X_1,X_2)$ \cite{schervish2012theory}.  Suppose  $g$ is orthogonal to $D$. For any $C_i\in\sigma(X_i)$ we have $C_i=X_i^{-1}(D_i)$, $D_i\in \mathcal{M}_i$, hence
\begin{align}
0=&E[g 1_{C_1}1_{C_2}]=E[\tilde g(X_1,X_2) 1_{C_1}1_{C_2}]\\
=&\int \tilde g(x_1,x_2) 1_{D_1}(x_1)1_{D_2}(x_2) dP_{(X_1,X_2)}.\notag
\end{align}
where $P_{(X_1,X_2)}$ is the distribution of $(X_1,X_2)$.  Therefore $\tilde g(x_1,x_2)dP_{(X_1,X_2)}$ is a complex measure on $\mathcal{M}_1\bigotimes\mathcal{M}_2$ that vanishes on all rectangles, hence it is the zero measure.  So $\tilde g=0$ $P_{(X_1,X_2)}$-a.s.  Hence $g=0$ $P$-a.s.

\end{proof}

\begin{lemma}\label{mart_rep_lemma}
Let $X\in L^2(\overline{\mathcal{G}^{W,\mathcal{C}}}_\infty,P)$.  Then there exists a unique $f\in L^2([0,\infty)\times \Omega,Prog,ds\times P,\mathbb{R}^n)$ ( $\overline{\mathcal{G}^{W,\mathcal{C}}}_t$-progressively measurable $\mathbb{R}^n$-valued $L^2$ functions) such that $X=E(X|\mathcal{C})+\int_{0}^\infty f_s dW_s$. We also have $E(X|\overline{\mathcal{G}^{W,\mathcal{C}}}_t)=E(X|\mathcal{C})+\int_{0}^t f_s dW_s$ for every $t$.
\end{lemma}
\begin{proof}

For any $f\in  L^2([0,\infty)\times \Omega,Prog,ds\times P,\mathbb{R}^n)$, $M_t=\int_{0}^t f_s dW_s$ is a continuous $L^2$-bounded martingale, hence it has an $L^2$ limit
\begin{align}
M_\infty=\int_{0}^\infty u_s dW_s\in  L^2(\overline{\mathcal{G}^{W,\mathcal{C}}}_\infty,P)
\end{align}
 and $M_t=E(M_\infty|\overline{\mathcal{G}^{W,\mathcal{C}}}_t)$.

If there exists two such function $f$, $\tilde f$ then $\int_{0}^\infty f_s-\tilde f_s dW_s=0$.  Using the It\^o isometry,
\begin{align}
0=&\left\|\int_{0}^\infty f_s-\tilde f_s dW_s\right\|^2_{L^2(P)}=\lim_{N\rightarrow\infty}E\left[\int_{0}^N \|f_s-\tilde f_s\|^2ds\right]\\
=&\|f-\tilde f\|_{L^2(ds\times P)}^2,\notag
\end{align}
so we have uniqueness.

Let $A$ be the set of $X\in L^2(\overline{\mathcal{G}^{W,\mathcal{C}}}_\infty,P)$ for which the claim holds. It is obviously a vector space. 

We show that $A$ is closed:\\
Let $X_j\in A$ with $X_j\rightarrow X$ in $L^2$.  Then there exists $f^j\in L^2(ds\times P,\mathbb{R}^n)$ s.t. $X_j=E(X_n|\mathcal{C})+\int_{0}^\infty f^j_s dW_s$.  By the It\^o isometry
\begin{align}
&\|X_j-X_l-E(X_j|\mathcal{C})+E(X_l|\mathcal{C})\|^2_{L^2(P)}=\|f^j-f^l\|^2_{L^2(ds\times P)}
\end{align}
and therefore
\begin{align}
\|f^j-f^l\|_{L^2(ds\times P)}\leq 2\|X_j-X_l\|_{L^2(P)}\rightarrow 0
\end{align}
as $j,l\rightarrow\infty$.  So $f^j$ is Cauchy.  By completeness, there exists $f\in L^2(ds\times P,\mathbb{R}^n)$ s.t. $f^j\rightarrow f$ in $L^2$.  Therefore
\begin{align}
X-E(X|\mathcal{C})=L^2-\lim_{j\to \infty}\int_{0}^\infty f^j_s dW_s.
\end{align}
Again by the It\^o isometry,
\begin{align}
\left\|\int_{0}^\infty f^j dW_s-\int_{0}^\infty f dW_s\right\|_{L^2(P)}=\|f^j-f\|_{L^2(ds\times P)}\rightarrow 0
\end{align}
so
\begin{align}
X-E(X|\mathcal{C})=\int_{0}^\infty f_s dW_s
\end{align}
which proves the claim.

If we can show that $A$ contains a  subset of  $L^2(\overline{\mathcal{G}^{W,\mathcal{C}}}_\infty,P)$ whose span is dense then we are done:\\
Given $Y\in L^2(\overline{\mathcal{F}^W}_\infty,P)$ the martingale representation for non-initially enlarged filtrations, i.e. for $\mathcal{C}$  the trivial sigma algebra, (see Section 3.4 in \cite{karatzas2014brownian}) gives $Y=E[Y]+\int_{0}^\infty h_s dW_s$ where $h\in L^2([0,\infty)\times \Omega,Prog,ds\times P,\mathbb{R}^n)$.  Therefore, for $C\in\mathcal{C}$ we have 
\begin{align}
1_C Y=1_CE[Y]+\int_{0}^\infty 1_Ch_s dW_s,
\end{align}
where we used the fact that $1_C$ is independent of $t$ and is $\overline{\mathcal{G}^{W,\mathcal{C}}}_{0}$-measurable.

  $1_Ch_s\in  L^2([0,\infty)\times \Omega,Prog,ds\times P,\mathbb{R}^n)$ and 
\begin{align}
&E(1_CY|\mathcal{C})=1_C E[Y]+ L^2-\lim_{N\rightarrow\infty} E\left(\left.\int_{0}^N 1_Ch_s dW_s\right|\mathcal{C}\right)\notag\\
=&1_C E[Y]+ \lim_{N\rightarrow\infty} E\left(\left.E\left(\left.\int_{0}^N 1_Ch_s dW_s\right|\overline{\mathcal{G}^{W,\mathcal{C}}}_{0}\right)\right|\mathcal{C}\right)\notag\\
=&1_C E[Y].\notag
\end{align}

This proves $1_CY\in A$.  In particular, $1_C1_B\in A$ for all $B\in \mathcal{F}^W_\infty$. Lemma \ref{density_lemma} applied to the random variables $Id:(\Omega,\mathcal{F})\rightarrow (\Omega,\mathcal{C})$ and $W:(\Omega,\mathcal{F})\rightarrow C([0,\infty),\mathbb{R}^n)$ proves that the span of such functions is dense  in $L^2(\overline{\mathcal{G}^{W,\mathcal{C}}}_\infty,P)$, thereby completing the proof.

\end{proof}

\subsection*{Acknowledgments}

J.W. was partially supported by NSF grants DMS 131271 and DMS 1615045.\\

\bibliographystyle{ieeetr}
\bibliography{refs}

\begin{thebibliography}{10}

\bibitem{smoluchowski1916drei}
M.~Smoluchowski, ``Drei vortrage uber diffusion, brownsche bewegung und
  koagulation von kolloidteilchen,'' {\em Zeitschrift fur Physik}, vol.~17,
  pp.~557--585, 1916.

\bibitem{KRAMERS1940284}
H.~Kramers, ``{Brownian motion in a field of force and the diffusion model of
  chemical reactions},'' {\em Physica}, vol.~7, no.~4, pp.~284 -- 304, 1940.

\bibitem{Nelson1967}
E.~Nelson, {\em Dynamical Theories of Brownian Motion}.
\newblock Mathematical Notes - Princeton University Press, Princeton University
  Press, 1967.

\bibitem{doi:10.1137/S1540345903421076}
G.~A. Pavliotis and A.~M. Stuart, ``White noise limits for inertial particles
  in a random field,'' {\em Multiscale Modeling \& Simulation}, vol.~1, no.~4,
  pp.~527--553, 2003.

\bibitem{Chevalier2008}
C.~Chevalier and F.~Debbasch, ``{Relativistic diffusions: A unifying
  approach},'' {\em Journal of Mathematical Physics}, vol.~49, no.~4, 2008.

\bibitem{bailleul2010stochastic}
I.~Bailleul, ``A stochastic approach to relativistic diffusions,'' in {\em
  Annales de l'institut Henri Poincar{\'e} (B)}, vol.~46, pp.~760--795, 2010.

\bibitem{pinsky1976isotropic}
M.~A. Pinsky, ``{Isotropic transport process on a Riemannian manifold},'' {\em
  Transactions of the American Mathematical Society}, vol.~218, pp.~353--360,
  1976.

\bibitem{pinsky1981homogenization}
M.~A. Pinsky, ``Homogenization in stochastic differential geometry,'' {\em
  Publications of the Research Institute for Mathematical Sciences}, vol.~17,
  no.~1, pp.~235--244, 1981.

\bibitem{Jorgensen1978}
E.~J{\o}rgensen, ``{Construction of the Brownian motion and the
  Ornstein-Uhlenbeck process in a Riemannian manifold on basis of the
  Gangolli-Mc.Kean injection scheme},'' {\em Zeitschrift f{\"u}r
  Wahrscheinlichkeitstheorie und Verwandte Gebiete}, vol.~44, no.~1,
  pp.~71--87, 1978.

\bibitem{dowell1980differentiable}
R.~M. Dowell, {\em Differentiable approximations to Brownian motion on
  manifolds}.
\newblock PhD thesis, University of Warwick, 1980.

\bibitem{XueMei2014}
X.-M. {Li}, ``{Random Perturbation to the Geodesic Equation},'' {\em Ann.
  Prob.}, vol.~44, no.~1, pp.~544--566, 2016.

\bibitem{angst2015kinetic}
J.~Angst, I.~Bailleul, and C.~Tardif, ``{Kinetic Brownian motion on Riemannian
  manifolds},'' {\em Electron. J. Probab.}, vol.~20, 2015.

\bibitem{bismut2005hypoelliptic}
J.-M. Bismut, ``{The hypoelliptic Laplacian on the cotangent bundle},'' {\em
  Journal of the American Mathematical Society}, vol.~18, no.~2, pp.~379--476,
  2005.

\bibitem{bismut2015}
J.-M. Bismut, ``{Hypoelliptic Laplacian and probability},'' {\em J. Math. Soc.
  Japan}, vol.~67, pp.~1317--1357, 10 2015.

\bibitem{PhysRevA.25.1130}
P.~H{\"a}nggi, ``Nonlinear fluctuations: The problem of deterministic limit and
  reconstruction of stochastic dynamics,'' {\em Phys. Rev. A}, vol.~25,
  pp.~1130--1136, Feb 1982.

\bibitem{Sancho1982}
J.~M. Sancho, M.~S. Miguel, and D.~D{\"u}rr, ``{Adiabatic elimination for
  systems of Brownian particles with nonconstant damping coefficients},'' {\em
  Journal of Statistical Physics}, vol.~28, no.~2, pp.~291--305, 1982.

\bibitem{volpe2010influence}
G.~Volpe, L.~Helden, T.~Brettschneider, J.~Wehr, and C.~Bechinger, ``Influence
  of noise on force measurements,'' {\em Physical review letters}, vol.~104,
  no.~17, p.~170602, 2010.

\bibitem{Hottovy2014}
S.~Hottovy, A.~McDaniel, G.~Volpe, and J.~Wehr, ``{The Smoluchowski-Kramers
  Limit of Stochastic Differential Equations with Arbitrary State-Dependent
  Friction},'' {\em Communications in Mathematical Physics}, vol.~336, no.~3,
  pp.~1259--1283, 2014.

\bibitem{herzog2015small}
D.~P. Herzog, S.~Hottovy, and G.~Volpe, ``{The small-mass limit for Langevin
  dynamics with unbounded coefficients and positive friction},'' {\em Journal
  of Statistical Physics}, vol.~163, no.~3, pp.~659--673, 2016.

\bibitem{particle_manifold_paper}
J.~Birrell, S.~Hottovy, G.~Volpe, and J.~Wehr, ``{Small Mass Limit of a
  Langevin Equation on a Manifold},'' {\em Annales Henri Poincar{\'e}},
  vol.~18, no.~2, pp.~707--755, 2017.

\bibitem{BirrellHomogenization}
J.~Birrell and J.~Wehr, ``{Homogenization of dissipative, noisy, Hamiltonian
  dynamics},'' {\em Stochastic Processes and their Applications}, 2017.

\bibitem{FritzRoughPath}
P.~Friz, P.~Gassiat, and T.~Lyons, ``{Physical Brownian motion in a magnetic
  field as a rough path},'' {\em Trans. Amer. Math. Soc.}, vol.~367,
  pp.~7939--7955, 2015.

\bibitem{fouque2007wave}
J.~Fouque, J.~Garnier, G.~Papanicolaou, and K.~Solna, {\em Wave Propagation and
  Time Reversal in Randomly Layered Media}.
\newblock Stochastic Modelling and Applied Probability, Springer New York,
  2007.

\bibitem{pavliotis2008multiscale}
G.~Pavliotis and A.~Stuart, {\em Multiscale Methods: Averaging and
  Homogenization}.
\newblock Texts in Applied Mathematics, Springer New York, 2008.

\bibitem{karatzas2014brownian}
I.~Karatzas and S.~Shreve, {\em Brownian Motion and Stochastic Calculus}.
\newblock Graduate Texts in Mathematics, Springer New York, 2014.

\bibitem{teschl2012ordinary}
G.~Teschl, {\em Ordinary Differential Equations and Dynamical Systems}.
\newblock Graduate studies in mathematics, American Mathematical Society, 2012.

\bibitem{Chetrite2008}
R.~Chetrite and K.~Gaw\c{e}dzki, ``{Fluctuation Relations for Diffusion
  Processes},'' {\em Communications in Mathematical Physics}, vol.~282, no.~2,
  pp.~469--518, 2008.

\bibitem{gawedzki2013fluctuation}
K.~Gaw\c{e}dzki, ``Fluctuation relations in stochastic thermodynamics,'' {\em
  arXiv preprint arXiv:1308.1518}, 2013.

\bibitem{Birrell_entropy}
J.~{Birrell}, ``{Entropy Anomaly in Langevin-Kramers Dynamics with Matrix Drag
  and Diffusion},'' {\em arXiv preprint arXiv:1709.06981}.

\bibitem{ortega2013matrix}
J.~Ortega, {\em Matrix Theory: A Second Course}.
\newblock University Series in Mathematics, Springer US, 2013.

\bibitem{schervish2012theory}
M.~Schervish, {\em Theory of Statistics}.
\newblock Springer Series in Statistics, Springer New York, 2012.

\end{thebibliography}

\end{document}